\documentclass[journal]{IEEEtran}
\usepackage[T1]{fontenc}
\usepackage[latin9]{inputenc}
\usepackage{amsthm}
\usepackage{amsmath} 
\usepackage{graphicx} 
\usepackage{color}
\usepackage[usenames,dvipsnames]{xcolor} 
\usepackage{cite}
\usepackage{algpseudocode}
\usepackage{algorithm}
\usepackage{amssymb}
\usepackage{subfigure}
\usepackage{esint}
\usepackage{algpseudocode}
\usepackage{epstopdf}
\usepackage{multirow}
\usepackage{makecell}
\usepackage{float}
\usepackage{lipsum}
\usepackage{balance}
\usepackage{ulem}
\usepackage{bbm}

\newtheorem{definition}{Definition}
\newtheorem{remark}{Remark}

\newtheorem{example}{Example}

\theoremstyle{plain}
\theoremstyle{plain}
\newtheorem{theorem}{Theorem}
\newtheorem{lemma}{Lemma}

\newcommand{\comment}[1]{}

\IEEEoverridecommandlockouts

\definecolor{green}{rgb}{0.13, 0.55, 0.13}
\definecolor{brown}{rgb}{0.6, 0.2, 0.0}

\begin{document}

\title{LOCO Codes Can Correct as Well: Error-Correction Constrained Coding for DNA Data Storage
}

\author{
   \IEEEauthorblockN{Canberk \.{I}rima\u{g}z{\i} and Ahmed Hareedy, \IEEEmembership{Member, IEEE}} 
   
   \thanks{This work was supported in part by the T\"{U}B\.ITAK 2232-B International Fellowship for Early Stage Researchers. 

Canberk \.{I}rima\u{g}z{\i} is with the Institute of Applied Mathematics (IAM), Middle East Technical University, 06800 Ankara, Turkey (e-mail: canberk.irimagzi@metu.edu.tr).

Ahmed Hareedy is with the Department of Electrical and Electronics Engineering, Middle East Technical University, 06800 Ankara, Turkey (e-mail: ahareedy@metu.edu.tr).}
}
\maketitle

\vspace{-5.0em}
\begin{abstract}
As a medium for cold data storage, DNA stands out as it promises significant gains in storage capacity and lifetime. However, it comes with its own data processing challenges to overcome. Constrained codes over the DNA alphabet $\{A,T,G,C\}$ have been used to design DNA sequences that are free of long homopolymers to increase stability, yet effective error detection and error correction are required to achieve reliability in data retrieval. Recently, we introduced lexicographically-ordered constrained (LOCO) codes, namely DNA LOCO (D-LOCO) codes, with error detection. In this paper, we equip our D-LOCO codes with error correction for substitution errors via syndrome-like decoding, designated as residue decoding. We only use D-LOCO codewords of indices divisible by a suitable redundancy metric $R(m) > 0$, where $m$ is the code length, for error correction. The idea is that the residue, which is the index modulo $R(m)$, of the received word index equals the residue of the index error, i.e., the index difference. We find an exhaustive list of index differences due to single-substitution errors, and we are able to recover the index of the original codeword from the residue of the received word index. This requires storing a table for index errors and their residues. Having this decoding algorithm in hand, we provide the community with a construction of constrained codes forbidding runs of length higher than fixed $\ell \in \{1,2,3\}$ and $GC$-content in $\big [0.5-\frac{1}{2K},0.5+\frac{1}{2K}\big ]$ that correct $K$ segmented substitution errors, one per codeword. We call the proposed codes error-correction (EC) D-LOCO codes. We also give a list-decoding procedure with near-quadratic time-complexity in $m$ to correct double-substitution errors within EC D-LOCO codewords, which has $> 98.20\%$ average success rate. The redundancy metric is projected to require $2\log_2(m)+O(1)$-bit allocation, i.e., $2\log_2(m)+O(1)$ reduction in message bits, for a length-$m$ codeword. Hence, our EC D-LOCO codes are projected to be capacity-approaching with respect to the error-free constrained system.

\end{abstract}


\begin{table*}[ht]
\vspace{-1em}
\caption{The Conceptual Correspondence Between Varshamov-Tenengolts Codes and Error-Correction DNA-LOCO Codes}
\centering
\scalebox{0.95}
{
\begin{tabular}{|c|c|c|c|c|} 
\hline
\shortstack{\\ Code} & \textrm{Supercode} & \textrm{Syndrome} & \textrm{Redundancy metric} & \textrm{Correction type} \\
 \hline
\shortstack{VT code \\ \textup{ }} & \shortstack{$\{0,1\}^m$ \\ \textup{ }} & \shortstack{$\sum_{i = 0}^{m-1}  i\cdot c_i$ \\ \textup{ }}& \shortstack{$m+1$ \\ \textup{ }} & \shortstack{\\ Single deletion or \\ single insertion} \\
 \hline
\shortstack{\\\\ EC D-LOCO code} & $\mathcal{D}_{m, \ell}$ & \hspace{0.15cm} \textrm{formal index} \hspace{0.15cm} & $\overset{\textup{ }}{\underset{\textup{ }}{R(m)}}$ & \textrm{Single substitution} \\
\hline
\end{tabular}}
\label{table:VT}
\vspace{-1em}
\end{table*}

\section{Introduction}\label{sec:intro}

DNA data storage is a promising candidate for the next-generation storage technology, which is required to address the needs posed by the ever-increasing volume of data produced in our world \cite{goldman}, \cite{grass_robust}, \cite{blawat}, \cite{erlich_fountain}, \cite{organick_etal}, \cite{wang_etal19}, \cite{milenkovic_rw2}, \cite{DNA_engine}, \cite{ding_etal}, \cite{DNAprint}. As a field rich with problems stemming from practical challenges, DNA data storage has attracted significant attention from coding theorists, where they can use their solid foundations to effectively tackle these problems and bring this technology to life. These challenges are primarily related to how to reliably perform writing (synthesis), storing, and reading (sequencing) information. One of these coding-theoretic problems \cite{cai_optimal}, \cite{nguyen_etal} is the following \cite{ross_etal}, \cite{schwartz_etal}:
\\
\textbf{Problem.} Construct error-correcting, balanced constrained codes over the DNA alphabet $\{A,T,G,C\}$ that 

\begin{itemize}
\item[(P1)] eliminate homopolymer runs of length more than certain minimum \cite{ross_etal} and 

\item[(P2)] correct the types of errors that occur in the DNA data storage channel, i.e., substitution, deletion, and insertion errors \cite{olgica_nano}, \cite{heckel_etal}.

\end{itemize}
Here, balance, in a less strict sense, refers to the property of DNA strands having global $GC$-content close to $50\%$. It is important to address this coding problem while offering affordable, complexity-wise, encoding-decoding algorithms and high code rates.\footnote{For all the rates in this paper, the unit is bits per nucleotide (DNA symbol).}

There are very few papers in the literature attempting to achieve all three objectives, no long homopolymer runs, $GC$ balance, and error correction, using the same class of codes. Nguyen et al. \cite{nguyen_etal} offer a construction of constrained codes (i) forbidding homopolymer runs of length $> 3$, (ii) with $GC$-content in $[0.5-\epsilon, 0.5+\epsilon]$, and (iii) that can correct a single edit (i.e., a single substitution, or a single deletion, or a single insertion). They encode $[2n-2(r_{\mathrm{RLL}} + r_{\epsilon} + 4)]$-bit binary data into a DNA sequence of length $n + 2\log_2(n)+ c$, where $r_{\mathrm{RLL}} = n - \lfloor \log_4(N) \rfloor$, $r_{\epsilon} = 2 \,\lceil \log_4(\lfloor 1/2\epsilon \rfloor +1) \rceil$, and $c$ is a constant independent of $n$. Here, $N$ is the cardinality (see \cite[Equation~(1)]{immink_cai}) of all $4$-ary length-$n$ constrained sequences forbidding runs of length $> 3$. Their encoding-decoding algorithms have quadratic time complexity in $n$ and low storage overhead (see Table~\ref{table:cmp}). Additionally, they offer linear-time encoding-decoding algorithms under a relaxed homopolymer constraint, allowing the maximum run-length to be $4$. In \cite{ding_etal}, Ding et al. show promising results of this construction when Illumina sequencing is adopted.

The literature has probabilistic error-correction constrained coding schemes for DNA data storage, where the guarantees on correction are provided by the average; see \cite{hedges}, \cite{wu_1}, \cite{dna_aeon}, \cite{wu_2}, \cite{wu_3}, \cite{parkpark} in addition to the first eight references. It also has error-correction non-constrained (or probabilistically constrained) coding schemes addressing (P2); see \cite{chandak_nano}, \cite{sima1}, \cite{sima2}, \cite{cai_seqrec}, \cite{song_delsub}, \cite{xing2del}, \cite{serge}, \cite{song_burst}, \cite{burst_ge}, and \cite{substring_farnoud}. Moreover, there are works offering balanced constrained codes addressing (P1); see \cite{wang_etal}, \cite{park_lee_no}, \cite{he_liu_wang_tang}, and our work \cite{D-LOCO}. The work in \cite{D-LOCO} is based on the theory of lexicographically-ordered constrained (LOCO) codes \cite{ahh_general}. A design of constrained codes preventing error-prone patterns for nanopore sequencing is offered in \cite{const_ONT}.

\subsection{Our contribution} \label{sec:contrib}

In this paper, we give a partial solution to the aforementioned problem in the introduction where we attempt to achieve all objectives with one class of codes, focusing only on substitution errors (most common of all three in (P2) for Illumina sequencing \cite{organick_etal}). Building on \cite{D-LOCO}, we  propose a systematic design of DNA sequences forbidding runs of length higher than fixed $\ell$ in $\{1, 2,3\}$ and having $GC$-content in $\big [0.5-\frac{1}{2K},0.5+\frac{1}{2K} \big]$ that correct $K$ segmented substitution errors, one per codeword. These DNA sequences are formed by concatenating $K$ error-correction (EC) DNA LOCO (D-LOCO) codewords, each of which has guaranteed correction capability of a single-substitution error. 

Our main contribution is to formally extend the encoding-decoding rule of D-LOCO codes to the set $\{A, T, G, C\}^m$. This rule satisfies a property that distinct D-LOCO codewords have distinct indices, though this is not the case for words $\in \{A, T, G, C\}^m$ in general (a weaker version of the confusability property in \cite{sima_syndrome}). Similarly to the novel technique of syndrome compression by Sima et al. \cite{sima_syndrome}, we search for minimal redundancy to guarantee the correction of single substitutions. Instead of appending syndromes to codewords for correction, we adopt a subcode perspective more suitable for the constrained setting. Via an effective use of bridging symbols, we remedy this lack of rule bijection for words $\in \{A, T, G, C\}^m$, allowing us to correct received words with forbidden patterns. Our construction fits into the framework of Varshamov-Tenengolts (VT) codes~\cite{VT}, which are the sets of length-$m$ binary sequences with associated integer syndromes of fixed residue modulo $m+1$ (see Table~\ref{table:VT}). Below are the major steps in our theoretical contribution:

\begin{enumerate}
\item \textit{Define the code.} We define an EC D-LOCO code as a subcode of the corresponding D-LOCO code, consisting of codewords whose indices are divisible by a suitable redundancy metric $R(m)>0$, where $m$ is the code length, for fixed $\ell$. Let the \textit{residue} of an integer be this integer modulo $R(m)$, and let the index of any received word be the result of passing this word by the D-LOCO decoding rule. For EC D-LOCO codewords, the residue of the received word index is equal to the residue of the index error, i.e., the difference between the received word index and the error-free codeword index. 

\item \textit{Find a working redundancy.} We theoretically find a set of index errors, exhausting all index differences due to single substitutions, for general $m$, and algorithmically search for the smallest redundancy metrics which guarantee the recoverability of index errors from their residues. 

\item \textit{Improve for lower storage and enhanced correction capability.} The residue decoding algorithm for EC D-LOCO codes relies on a table of index errors and their residues and is designed in order to address double-substitution errors per segment as well with high detection/correction success rate (with probability $99.95\%/98.20\%$) at no rate cost. The algorithm runs in quadratic (at worst cubic) time for single-substitution (double-substitution) errors per segment, and it requires a storage overhead of less than $6$ kilobytes (kB) for moderate-to-high code lengths after an optimization.\footnote{We adopt the convention that a kilobyte (kB) is $8000$ bits.}
\end{enumerate}

EC D-LOCO codes are expected to be capacity-approaching as they incur only $2\log_2(m)+O(1)$-bit subcode redundancy due to $R(m)$ in the message length (see Remark~\ref{rem:redundancy}). They also offer low error-propagation, reconfigurability, and local balancing, which are intrinsic to all LOCO codes (see \cite[Section~VII C--E]{D-LOCO}). The first error-correction idea for LOCO codes appeared in~\cite{ahh_loco}, where the codewords of some binary LOCO code, encoding parity bits of a low-density parity-check (LDPC) code, are corrected in case a forbidden pattern shows up after transmission. This is done by first (i) flipping a suitably chosen bit of a received word, and then (ii) feeding the adjusted word to 
an LDPC decoder along with the bits of the input data. In this paper, we offer the first systematic framework to do that, which opens the door for similar efforts in other applications, regardless of the type of the data encoded and whether received words are free of forbidden patterns or not. In particular, residue decoding applies to any LOCO codes (see \cite{ahh_loco}, \cite{bd_tdmr}, \cite{ahh_mdl}, \cite{ahh_ps}, \cite{dua_tdmr}, \cite{goca_tdmr}) designed via the general method in \cite{ahh_general}, equipping them with substitution error correction for various applications.

The rest of the paper is organized as follows. In Section~\ref{sec:prelim}, we recall the encoding-decoding rule of D-LOCO codes that returns the index of a given codeword. We also give an illustration of how indices are affected by single-substitution errors. Index errors are studied in Section~\ref{sec:index_error_ell1} for $\ell=1$ and in Section~\ref{sec:index_error_ell2} for $\ell=2$, then they are discussed in Section~\ref{sec:model} for $\ell=3$. In Section~\ref{sec:residue_algorithm}, we introduce EC D-LOCO codes and present the residue decoding algorithm, illustrated
with examples and a flowchart. In Section~\ref{sec:storage-overhead}, we optimize the storage overhead via standard techniques and discuss the predicted capacity-approachability of EC D-LOCO codes based on a polynomial model of the redundancy metric $R(m)$. Next, we compare our work with \cite{nguyen_etal} and \cite{barlev_etal} from the finite-length perspective and with \cite{song_delsub} in terms of asymptotics, and then discuss how to address double-substitution errors in Section~\ref{sec:comp_doub}. Finally, we end the paper with conclusion and future directions.

\section{Preliminaries} \label{sec:prelim}
In \cite{D-LOCO}, we defined D-LOCO codes addressing (P1) as well as the balance constraint. Using the novel technique of LOCO coding \cite{ahh_general}, we obtained the encoding-decoding rule for D-LOCO codes, whose encoding-decoding algorithms are based on \cite{D-LOCO}. We state this rule below for reference in the paper.

\begin{definition} An element $\mathbf{w}$ in $\{A, T, G, C\}^m$, for some $m \geq 1$, is called a word (of length $m$).
\end{definition}

\begin{definition} \label{def:LOCO} The D-LOCO code $\mathcal{D}_{m,\ell}$ with parameters $m \geq \ell \geq 1$ is defined as the set of all words of length $m$ that do not contain any patterns in $\{\mathbf{\Lambda}^{\ell+1} \textrm{ } | \textrm{ } \Lambda \in \{A,T,G,C\}\}$. Here, $\mathbf{\Lambda}^{\ell+1}$ is the sequence of length $\ell+1$ all of whose symbols are $\Lambda$, and we call such a sequence $\mathbf{\Lambda}^{\ell+1}$ a run of length $\ell+1$. 

Elements of $\mathcal{D}_{m,\ell}$ are called codewords. They are ordered lexicographically in an ascending manner according to the rule $A < T < G < C$ for any symbol, where the symbol significance reduces from left to right.
\end{definition}

\begin{definition} \label{def:card} The codebook size of the D-LOCO code $\mathcal{D}_{m,\ell}$ is denoted by $N_{\mathrm{D}}(m,\ell)$. For ease of notation, we denote this cardinality by $N(m)$ whenever the context reveals $\ell$. We have $N(m) = 3 \sum^{\ell}_{k=1} N(m-k)$, where $N(0) \triangleq 4/3$ and $N(r) \triangleq 0$ for $r \leq -1$ \cite[Equation (1)]{immink_cai}.
\end{definition}

\begin{remark} \label{rem:word} When a substitution error occurs in the codeword $\mathbf{c} \in \mathcal{D}_{m,\ell}$, the resulting erroneous word $\mathbf{w}$ may not be free of runs of length $\ell + 1$. However, it is guaranteed that $\mathbf{w}$ contains at most one run of largest length $\in \{\ell+1, \ell+2, \cdots, 2\ell+1\}$. The longest erroneous run results from an error on the symbol between two runs of length $\ell$ each.
\end{remark}

The encoding-decoding rule gives the index of a D-LOCO codeword given its own symbols (decoding) and vice versa (encoding). Let $g(\mathbf{c})$ be the index of the D-LOCO codeword $\mathbf{c} \triangleq c_{m-1} c_{m-2} \dots c_0$ in $\mathcal{D}_{m,\ell}$. Below, we let $\delta$ stand for the small letter in $\{\mathsf{a},\mathsf{t},\mathsf{g}\}$ corresponding to some $\Delta$ in $\{A, T, G\}$.

\begin{theorem} \label{thm:encoding_rule} The encoding-decoding rule $g: \mathcal{D}_{m,\ell} \rightarrow \{0,1,\break\dots,N(m)-1\}$, for general $\ell \geq 1$, is as follows:
\begin{equation}
g(\mathbf{c})= \frac{3}{4} \sum^{m-1}_{i=0} \sum^{\ell}_{j=1} \sum^{j}_{k=1} (\mathsf{a}_{i,k}+\mathsf{t}_{i,k}+\mathsf{g}_{i,k}) N(i+j-\ell). \label{eq:rule_gen}
\end{equation}

\noindent For $1 \leq k \leq \ell$, $\delta_{i,k} \in \{0,1\}$ and $\delta_ {i,k} = 1$ if and only if $c_i > \Delta, c_{i+j} = \Delta$ for all $i+1 \leq j \leq i+k-1$, and $c_{i+k} \neq \Delta$. 

Note that for $k = 1$, the condition $c_{i+j} = \Delta$ for all $i+1 \leq j \leq i+k-1$ is naturally satisfied because there is no $j$ satisfying $i+1 \leq j \leq i+k-1$ for $k=1$. 

Also, we define $c_i \triangleq C$, for all $i > m-1$, to represent ``out of codeword bounds''. 
\end{theorem}
\begin{proof}
In \cite{ahh_general}, $g_i(c_i)$ is formulated as follows:
\begin{equation} g_i(c_i)=\sum_{c_i' < c_i} N_{\mathrm{symb}}(m,c_{m-1}c_{m-2}\dots c_{i+1} c_i'), 
\end{equation}
where $N_{\mathrm{symb}}(m,c_{m-1}c_{m-2}\dots c_{i+1} c_i')$ computes the number of codewords $\mathbf{u}=c_i'u_{i-1}\dots u_1u_0$ in $\mathcal{D}_{i+1,\ell}$ such that $c_{m-1}c_m \dots c_{i+1}c_i'u_{i-1}\dots u_1u_0$ is a codeword in $\mathcal{D}_{m,\ell}$. Note that $g_i(A)=0$ and letting $N_{A,k}(n)$ denote the number of length-$n$ codewords starting with exactly $k$  consecutive $A$'s on the left, we have
\begin{align} &g_i(T) = N_{\mathrm{symb}}(m,c_{m-1}c_{m-2}\dots c_{i+1}A) \nonumber \\
&=\mathsf{a}_{i,1}[N_{A,1}(i+1)+N_{A,2}(i+1) +\dots +N_{A,\ell}(i+1)] \nonumber \\
&+\mathsf{a}_{i,2}[N_{A,1}(i+1)+N_{A,2}(i+1)+\dots+N_{A,\ell-1}(i+1)] \nonumber \\
&+\dots \nonumber \\
&+\mathsf{a}_{i,\ell-1}[N_{A,1}(i+1)+N_{A,2}(i+1)]+\mathsf{a}_{i,\ell}N_{A,1}(i+1) \nonumber \\
&= (\mathsf{a}_{i,1}+\mathsf{a}_{i,2}+\dots+\mathsf{a}_{i,\ell})N_{A,1}(i+1) + \dots \nonumber \\
&+(\mathsf{a}_{i,1}+\mathsf{a}_{i,2})N_{A,\ell-1}(i+1)+\mathsf{a}_{i,1}N_{A,\ell}(i+1) \nonumber \\ 
&=\frac{3}{4} \sum^{\ell}_{j=1} \sum^{j}_{k=1} \mathsf{a}_{i,k} N(i+j-\ell),
\end{align}
where $\mathsf{a}_{i,k}$'s are as in the statement. $g_i(G)$ and $g_i(C)$ can be obtained in a similar manner, and the proof completes.

\end{proof}
\begin{figure*}[ht!]
\vspace{-1.5em}
\noindent\makebox[\linewidth]{\rule{\textwidth}{1,25pt}}
\begin{equation} \label{eq:explicit_exp_ell=1} g(\mathbf{c}) =  \frac{3}{4}\sum^{m-1}_{i=0} (\mathsf{a}_{i,1}+\mathsf{t}_{i,1}+\mathsf{g}_{i,1}) N(i)  = \frac{1}{4}\sum^{m-1}_{i=0} (\mathsf{a}_{i,1}+\mathsf{t}_{i,1}+\mathsf{g}_{i,1}) N(i+1), \hspace{1cm} \textrm{ where $\mathbf{c} \in \mathcal{D}_{m, 1}$.}
\end{equation}
\noindent\makebox[\linewidth]{\rule{\textwidth}{1,25pt}}
\begin{align} \label{eq:explicit_exp_ell=2} g(\mathbf{c}) &=  \frac{3}{4}\sum^{m-1}_{i=0} \bigg[ (\mathsf{a}_{i,1}+\mathsf{t}_{i,1}+\mathsf{g}_{i,1}+\mathsf{a}_{i,2}+ \mathsf{t}_{i,2}+\mathsf{g}_{i,2}) N(i)+ (\mathsf{a}_{i,1}+ \mathsf{t}_{i,1}+\mathsf{g}_{i,1})N(i-1)   \bigg]
\nonumber  \\ 
&= \frac{1}{4}\sum^{m-1}_{i=0} \bigg[ (\mathsf{a}_{i,1}+\mathsf{t}_{i,1}+\mathsf{g}_{i,1}) N(i+1)+ 3\cdot(\mathsf{a}_{i,2}+ \mathsf{t}_{i,2}+\mathsf{g}_{i,2})N(i)   \bigg], \hspace{1cm} \textrm{ where $\mathbf{c} \in \mathcal{D}_{m, 2}$.} 
\end{align}
\noindent\makebox[\linewidth]{\rule{\textwidth}{1,25pt}}
\vspace{-2em}
\end{figure*}

\begin{definition} 
\noindent For a codeword $\mathbf{c} \in \mathcal{D}_{m, \ell}$, the inner sums
\begin{equation} g_i(c_i) \triangleq \frac{3}{4} \sum^{\ell}_{j=1} \sum^{j}_{k=1} (\mathsf{a}_{i,k}+\mathsf{t}_{i,k}+\mathsf{g}_{i,k}) N(i+j-\ell) 
\end{equation}
 in (\ref{eq:rule_gen})  
 are collectively called the symbol contribution of the $i^{\textrm{th}}$ symbol $c_i$ to the index of $\mathbf{c}$, where $g_i(c_i)$ is a brief version of $g_i(\mathbf{c},c_i)$ and we have $g(\mathbf{c})= \sum^{m-1}_{i=0} g_i(c_i)$ \cite{ahh_general}. 
\end{definition}

Based on (2), we define a notion of depth to keep track of the changes in symbol contributions due to substitutions. 

\begin{definition} \label{def:depth} Fix $\mathbf{c} \in \mathcal{D}_{m,\ell}$. Given a letter $\Delta \in \{A, T, G\}$ with $c_i > \Delta$, the \textit{depth} of $\Delta$ at $i$ is defined to be 
\[ \mathrm{dep}_i(\Delta) =\begin{cases} k& \textrm{if } \delta_{i,k} = 1 \textrm{ for unique $k$ s.t. }  1 \leq k \leq \ell, \\
\ell + 1 & \textrm{otherwise} \\ & \textrm{(i.e., if } \delta_{i,k} = 0 \textrm{ for all } 1 \leq k \leq \ell).
\end{cases} 
\]
In words, the depth of a symbol $\Delta$ smaller than $c_i$ at $i$ is ``(how many times it appears right before $c_i$ consecutively) $+ 1$''. We set $\mathrm{dep}_i(\Delta) \triangleq 0$ if $c_i \leq \Delta \leq C$ for any $0 \leq i \leq m-1$. 

We use the shorthand notation $\mathrm{dep}_i(\Delta): k_1 \rightarrow k_2$ to denote the depths of $\Delta  \in \{A,T,G,C\}$ before and after a substitution error at some symbol $c_j$ of $\mathbf{c}$, where $j$ may or may not be the same as $i$.
More precisely, the notation means that the depth of $\Delta$ at $i$ in the codeword $\mathbf{c}$ is $k_1$, and after $c_j$ is substituted with a symbol $w_j = c_j' \neq c_j$, the depth of $\Delta$ at $i$ in the new word $\mathbf{w}$ becomes $k_2$, where $\mathbf{w}\triangleq w_{m-1}w_{m-2}\dots w_0=c_{m-1}c_{m-2}\dots c_{j+1}c_j' c_{j-1}\dots  c_0$. We will refer to $k_1$ ($k_2$) as the initial (final) depth of $\Delta$. 
\end{definition}

For readers' convenience and future reference, we give an illustration of  how the index of a codeword is computed using Theorem~\ref{thm:encoding_rule}. Note that (\ref{eq:explicit_exp_ell=1}) and (\ref{eq:explicit_exp_ell=2}) give alternative expressions for the encoding-decoding rule where $\ell = 1$ and $\ell=2$, respectively. They can be obtained from (\ref{eq:rule_gen}) by using the recursive relation in Definition~\ref{def:card}.

\begin{example} \label{exam:index_comp} 
Consider the encoding-decoding rule $g:\mathcal{D}_{6,1} \rightarrow \{0,1,2,\cdots,971\},$
and the codeword $\mathbf{c} = AGTCAG \in \mathcal{D}_{6,1}$. The cardinality of $\mathcal{D}_{m,1}$ obeys the recursive relation $N_{\mathrm{D}}(m,1)=3N_{\mathrm{D}}(m-1,1)$, or in short $N(m) = 3N(m-1)$,
where $N(0) \triangleq 4/3$. In particular, $N(1)=4$, $N(2)=12$, $N(3)=36$, $N(4)=108$, $N(5)= 324$, and $N(6)=972$.  We have:

\begin{itemize}
\item[a)] $g_5(A)=g_1(A)=0$. 
\item[b)] For $c_4$, $\mathsf{t}_{4,1}=1$ and all other coefficients are zero so that $g_4(G)=N(5)/4=81$.
\item[c)] For $c_3$, $\mathsf{a}_{3,1}=1$ and all other coefficients are zero so that $g_3(T)=N(4)/4=27$.
\item[d)] For $c_2$, $\mathsf{a}_{2,1}=\mathsf{g}_{2,1}=1$ and all other coefficients are zero so that $g_2(C)=2 \cdot N(3)/4=18$.
\item[e)] For $c_0$, $\mathsf{t}_{0,1}=1$ and all other coefficients are zero so that $g_0(G)=N(1)/4=1$.
\end{itemize}
Adding these, we get $g(\mathbf{c})=\sum^{5}_{i=0} g_i(c_i)=127$. 
\end{example}

\begin{remark} \label{rem:formal_index}
Note that the encoding-decoding rule $g$ on $\mathcal{D}_{m,\ell}$ extends to all words of length $m$, where the range becomes bigger. However, mapping here is not one-to-one and not necessarily onto. For instance, let us fix $\ell = 1$ and consider the word $\mathbf{w}=AG\mathbf{C}CAG \not \in \mathcal{D}_{6,1}$. We have
\begin{itemize}
\item[a)] $g_5(A)=g_1(A) = 0$. 
\item[b)] For $w_4$, $\mathsf{t}_{4,1}=1$ and all other coefficients are zero so that $g_4(G)=N(5)/4=81$.
\item[c)] For $w_3$, $\mathsf{a}_{3,1}=\mathsf{t}_{3,1}=1$ and all other coefficients are zero so that $g_3(C)=2 \cdot N(4)/4=54$.
\item[d)] For $w_2$, $\mathsf{a}_{2,1}=\mathsf{t}_{2,1}=\mathsf{g}_{2,1}=1$ and all other coefficients are zero so that $g_2(C)=3 \cdot N(3)/4=27$.
\item[e)] For $w_0$, $\mathsf{a}_{0,1}=1$ and all other coefficients are zero so that $g_0(G)=N(1)/4=1$.
\end{itemize}
These contributions are denoted by $g_i(\mathbf{w}, w_i, \ell)$, which is abbreviated to $g_i(\mathbf{w}, w_i)$ or to $g_i(w_i)$ when the context is clear. Adding these, we get $\sum^{5}_{i=0} g_i(w_i)=163$. 

In the context where $\ell$ is fixed and for a word $\mathbf{w} \in \{A,T,G,C\}^m$, the number $\sum^{m-1}_{i=0} g_i(w_i)$ will be called its formal index and we will denote it by $g(\mathbf{w}, \ell)$. For codewords in $\mathcal{D}_{m,\ell}$, we have $g(\cdot, \ell)=g(\cdot)$. Note that the formal index of a word $\mathbf{w}$ can exceed $N_{\mathrm{D}}(m,\ell)-1$. 
\end{remark}

\begin{remark} \label{rem:interdependence} Comparing the index computations for $\mathbf{c}$ in Example~\ref{exam:index_comp} and $\mathbf{w}$ in Remark~\ref{rem:formal_index} above, we observe that a substitution error at $c_3$ not only affects the symbol contribution $g_3(\cdot)$ but also $g_2(\cdot)$ as a result of $\mathrm{dep}_2(c_3 = T): 2 \rightarrow 1$. Due to this phenomenon, the number of index differences due to single substitutions increases with $\ell$, resulting in higher redundancy metric for fixed $m$. We will answer the question of whether, given the finite-length rate considerations, it is worth it to use $\mathcal{D}_{m,3}$ over $\mathcal{D}_{m,2}$ in designing DNA strands for practical code lengths (see Fig.~\ref{fig:cmp} in Section~\ref{sec:model}).
\end{remark}

\begin{definition} \label{def:set_difference} For D-LOCO code $\mathcal{D}_{m,\ell}$, we denote the set 
\begin{align*} \big \{g(\mathbf{w},\ell)-g(\mathbf{c}) \textrm{ } | &\textrm{ } \mathbf{c} \in \mathcal{D}_{m,\ell}, \textup{ }  \mathbf{w} \in \{A,T,G,C\}^m, \textrm{ and } \nonumber \\
& \mathbf{c} \textrm{ and } \mathbf{w} \textrm{ differ at at most one location}\big\}
\end{align*}
of index differences due to at most one substitution error by $\mathcal{E}(m, \ell)$, and $\mathcal{E}^+(m, \ell)$ will denote the set of non-negative index differences in $\mathcal{E}(m, \ell)$.
\end{definition}
We will tackle $\mathcal{E}^+(m,\ell)$ in the next two sections for $\ell \in \{1,2\}$. The case of $\ell=1$ will serve as a prototype, while the case of $\ell=2$ is more involved and 
will be of utmost interest to us as we will illustrate. 

\section{Index Errors Due to Single Substitution \\ for $\ell =1$} \label{sec:index_error_ell1}

\subsection{Index Differences for $\ell =1$} \label{sec:index-diff-ell=1}

For $\ell=1$, the (minimal) set of forbidden patterns includes only $2$-tuple patterns (see Definition~\ref{def:LOCO}). Therefore, a substitution error can only affect two symbol contributions. Consequently, an index difference in $\mathcal{E}^+(m,1)$ is of the form
\begin{align} g_i(\mathbf{w},w_i) - g_i(\mathbf{c},c_i) + \big[g_{i-1}(\mathbf{w},c_{i-1})-g_{i-1}(\mathbf{c},c_{i-1})\big],
\end{align}
where $w_i > c_i$ for some $0 \leq i \leq m-1$ and $w_j=c_j$ for all other locations $j \neq i$. Here, we set $g_{-1}(\cdot) \triangleq 0$ and denote
\vspace{-0.1em}
\begin{align}
\hspace{-0.5em}g_i(c_{i+1} c_{i} \rightarrow c_{i+1} w_i) & \triangleq g_i(\mathbf{w},w_i) - g_i(\mathbf{c},c_i), \nonumber \\
\hspace{-0.5em}g_{i-1}(c_i c_{i-1} \rightarrow w_i c_{i-1}) & \triangleq  g_{i-1}(\mathbf{w},c_{i-1})-g_{i-1}(\mathbf{c},c_{i-1}).
\end{align}
In case we want to skip the symbols for brevity, we will simply denote these differences by $g_i^{\Delta}$ and $g_{i-1}^{\Delta}$, respectively. We study possible values for the changes in symbol contributions individually below, and then discuss their interdependence in Remark~\ref{rem:interdependence}. We follow the correspondence ($c_i \longleftrightarrow a_i$) below: 
\begin{align} \label{eqn_int} A &\longleftrightarrow 0 \textup{ } (\mathrm{mod} \textup{ } 4), \hspace{1cm} T \longleftrightarrow 1 \textup{ } (\mathrm{mod} \textup{ } 4), \nonumber \\
G &\longleftrightarrow 2 \textup{ } (\mathrm{mod} \textup{ } 4), \hspace{1cm} C \longleftrightarrow 3 \textup{ } (\mathrm{mod} \textup{ } 4).
\end{align}
The integer $(\mathrm{mod} \textup{ } 4)$ equivalent of symbol $c_i$ is defined according to (\ref{eqn_int}) as $a_i$. For simplicity, while expressing $g_i^{\Delta}$ and $g_{i-1}^{\Delta}$ in terms of symbols $c_i$ and $w_i$, we overload the notation $c_i$ and $w_i$ with the respective integer equivalents. Note that $A$ and $C$ as well as $T$ and $G$ are said to be complements. Codeword complements are always symbol-wise complements.

\subsection{Changes in Symbol Contributions for $\ell =1$}

We study the above changes $g_i(c_{i+1} c_i \rightarrow c_{i+1} w_i)$ and $g_{i-1}(c_i c_{i-1} \rightarrow w_i c_{i-1})$ using the ordering of the symbols $c_i$, $w_i$, and $c_{i+1}$, where $c_i < w_i$. Note that $c_{i+1}$ needs~to be different from $c_i$ as the original codeword in $\mathcal{D}_{m,1}$ is free of runs of length $2$. Note, however, that after a substitution error, the received word may not be a D-LOCO codeword (see Remark~\ref{rem:word}). For the following analysis, we assume $i \geq 1$. 

We first study $g_i(c_{i+1} c_i \rightarrow c_{i+1} w_i)$.
\begin{enumerate}
\item Suppose $c_i < w_i \leq c_{i+1}$, which implies that there is no change in the depth $\mathrm{dep}_i(c_{i+1})$ (and it is $0$). Thus,
$$g_i^{\Delta} = \frac{1}{4}(w_i-c_i)N(i+1) \textrm{, where } w_i- c_i \in \{1,2,3\}.$$

\item Suppose $c_i < c_{i+1} < w_i$, which implies that the depth $\mathrm{dep}_i(c_{i+1})$ changes from $0$ to $2$. Thus,
$$g_i^{\Delta} = \frac{1}{4}(w_i-c_i-1)N(i+1) \textrm{, where } w_i- c_i \in \{2,3\}.$$

\item Suppose $c_{i+1} < c_i < w_i $, which implies that there is no change in the depth $\mathrm{dep}_i(c_{i+1})$ (and it is $2$). Thus,
$$g_i^{\Delta} = \frac{1}{4}(w_i-c_i)N(i+1) \textrm{, where } w_i- c_i \in \{1,2\}.$$
\end{enumerate}
	
We now study $g_{i-1}(c_i c_{i-1} \rightarrow w_i c_{i-1})$.		
Note that if $c_{i-1} < c_i$, then $g_{i-1}(c_i c_{i-1} \rightarrow w_i c_{i-1})=0$ since $c_i < w_i$. Thus, a nonzero change in symbol contribution at $i-1$ can occur, if any, only due to changes in 
 \begin{enumerate}
 \item the depth of $c_i$ at $i-1$ if $c_i < c_{i-1} \leq  w_i$, or
 \item the depths of $c_i$ and $w_i$ at $i-1$ if $c_i < w_i < c_{i-1}$.
 \end{enumerate}
 
\noindent In the first case, $\mathrm{dep}_{i-1}(c_i): 2 \rightarrow 1$ (one $\delta_{i-1,1}$ becomes $1$), and thus using (\ref{eq:explicit_exp_ell=1})
$$g_{i-1}(c_i c_{i-1} \rightarrow w_i c_{i-1}) = \frac{1}{4}N(i).$$ In the second case, we must have $\mathrm{dep}_{i-1}(c_i): 2 \rightarrow 1$ and $\mathrm{dep}_{i-1}(w_i): 1 \rightarrow 2$, and thus
$$g_{i-1}(c_i c_{i-1} \rightarrow w_i c_{i-1}) = 0.$$

\begin{example} 
For the case of $c_{i+1} c_i c_{i-1} = TAC$ and $w_i = G$, we have $\mathrm{dep}_i(c_{i+1} = T): 0 \rightarrow 2$, leading to $g_i(TA \rightarrow TG) = N(i+1)/4$. Moreover, $\mathrm{dep}_{i-1}(c_i = A): 2 \rightarrow 1$ and $\mathrm{dep}_{i-1}(w_i = G): 1 \rightarrow 2$, leading to $g_{i-1}(AC \rightarrow GC) = 0$. Hence, the index difference $g(\mathbf{w},1)-g(\mathbf{c})$ is $N(i+1)/4$. Observe that the number of codewords starting with $A$ or $G$ is the same because of code symmetry.
\end{example}

\begin{definition} \label{def:indexset1} We denote the set 
$$\big \{ \big[i; \theta \cdot N(i+1)/4 \big] \textrm{ } | \textrm{ } 0 \leq i \leq m-1, \theta \in \{1,2,3\}\big \}$$
of pairs by $\mathsf{E}_1$, and the set
$$\big \{ \big[i-1; \theta \cdot N(i)/4 \big] \textrm{ } | \textrm{ } 1 \leq i \leq m-1, \theta \in \{0, 1\}\big \} \, \cup \, \{[-1;0]\}$$
of pairs by $\mathsf{E}_2$. Here, $i$ is the location of the substitution.
\end{definition}

Below, we introduce a remark that will be more crucial when $\ell$ is in $\{2,3\}$ (see Section~\ref{sec:interdependence_for_ell=2}).

\begin{remark} \label{rem:interdependence_opt}
Note that $g_i^{\Delta}$ can take the expression $\frac{3}{4}N(i+1)$ only if $c_i = A$ and $w_i=C$, and in this scenario, for $i \geq 1$,
\begin{itemize} 
\item[-] the case $c_{i-1} < c_i$, and
\item[-]  the second case in the analysis of $g_{i-1}(c_i c_{i-1} \rightarrow w_i c_{i-1})$ 
\end{itemize} 
cannot be realized. This is an instance showing the fact that changes in symbol contributions at indices $i$ and $i-1$ are not independent from each other. Practically, the set 
\begin{align*}
\{x + y \textrm{ } | \textrm{ } [i; x] \in \mathsf{E}_1, \textrm{ } [i-1;y] \in \mathsf{E}_2 \textrm{ for all } 0 \leq i \leq m-1\} 
\end{align*}
is larger than what we want. It is important to eliminate as many cases as possible while bounding $\mathcal{E}^+(m,\ell)$ from above in order to achieve higher rates and lower storage overhead.
\end{remark}

\begin{definition} \label{def:indexset1} The set 
\begin{align*}  & \bigg \{\theta_1   \cdot  N(i+ 1)  /4   + \theta_2 \cdot N(i)/4 \textrm{ } \big | \textrm{ }  1 \leq i \leq m-1,   \\
& \textrm{ } (\theta_1,\theta_2) \in \big \{(1,0),(2,0), (1,1), (2,1),(3,1)\big\} \bigg \} \, \cup \, \{0,1,2, 3\} 
\end{align*}
is defined as the set of non-negative index errors for $\ell=1$ and denoted by $\mathcal{E}_{\textrm{sup}}^+(m,1)$. We have $\mathcal{E}^+(m,1) \subseteq \mathcal{E}_{\textrm{sup}}^+(m,1)$, and the cardinality $|\mathcal{E}_{\textrm{sup}}^+(m,1)|$ is around $5m$.
\end{definition}

\section{Index Errors Due to Single Substitution \\ for $\ell=2$} \label{sec:index_error_ell2}
\subsection{Index Differences for $\ell =2$} \label{sec:index_diff_ell=2}

For $\ell=2$, the (minimal) set of forbidden patterns includes only $3$-tuple patterns (see Definition~\ref{def:LOCO}). Therefore, a substitution error can only affect three symbol contributions. Consequently, an index difference in $\mathcal{E}^+(m,2)$ is of the form
\begin{align}  g_i(\mathbf{w},w_i) &- g_i(\mathbf{c},c_i) + \big[g_{i-1}(\mathbf{w},c_{i-1})-g_{i-1}(\mathbf{c},c_{i-1})\big] \nonumber \\
&+ \big[g_{i-2}(\mathbf{w},c_{i-2})-g_{i-2}(\mathbf{c},c_{i-2})\big],
\end{align}
where $w_i > c_i$ for some $0 \leq i \leq m-1$ and $w_j=c_j$ for all other locations $j \neq i$. Here, we set $g_{-1}(\cdot)=g_{-2}(\cdot)=0$. For simplicity, we denote 
\begin{align}\label{diffs_l2}
g_i(c_{i+2} c_{i+1} &c_i  \rightarrow c_{i+2}c_{i+1} w_i)  \triangleq g_i(\mathbf{w},w_i) - g_i(\mathbf{c},c_i), \nonumber \\
g_{i-1}(c_{i+1}c_i & c_{i-1} \rightarrow c_{i+1}w_i c_{i-1}) \nonumber \\
& \triangleq  g_{i-1}(\mathbf{w},c_{i-1})-g_{i-1}(\mathbf{c},c_{i-1}), \textrm{ and} \nonumber \\
g_{i-2} (c_i c_{i-1} & c_{i-2} \rightarrow w_i c_{i-1} c_{i-2}) \nonumber\\
&\triangleq g_{i-2}(\mathbf{w},c_{i-2})-g_{i-2}(\mathbf{c},c_{i-2}).
\end{align}

We study possible values for changes in symbol contributions individually below. In case we want to skip the symbols, we simply denote the differences in (\ref{diffs_l2}) by $g_i^{\Delta}$, $g_{i-1}^{\Delta}$, and $g_{i-2}^{\Delta}$. These values are interdependent, and we will revisit them in Section~\ref{sec:interdependence_for_ell=2}.

\subsection{Changes in Symbol Contributions for $\ell =2$}
In this section, we obtain the possible values for $g_i^{\Delta}$, $g_{i-1}^{\Delta}$, and $g_{i-2}^{\Delta}$ based on the ordering of the symbols $c_i , w_i$, and $c_{i+1}$ where $c_i < w_i$. Note that $c_i$, $c_{i+1}$, and $c_{i+2}$ cannot all be the same symbol as the original codeword in $\mathcal{D}_{m,2}$ is free of runs of length $3$. For the following analysis, we assume $i \geq 2$. 

We first study $g_i(c_{i+2} c_{i+1} c_i \rightarrow c_{i+2} c_{i+1} w_i)$ using (\ref{eq:explicit_exp_ell=2}).

\begin{enumerate}
\item Suppose $c_i < w_i \leq c_{i+1}$, which implies that there is no change in the depth $\mathrm{dep}_i(c_{i+1})$ (and it is $0$). Thus,
$$g_i^{\Delta} = \frac{1}{4}(w_i-c_i)N(i+1) \textrm{ where } w_i- c_i \in \{1,2,3\}.$$

\item Suppose $c_i \leq c_{i+1} < w_i $ and $c_{i+2} \neq c_{i+1}$, which implies that $\mathrm{dep}_i(c_{i+1}): 0 \rightarrow 2$. Thus, \\
$$g_i^{\Delta} = \frac{1}{4}(w_i-c_i-1)N(i+1)+ \frac{3}{4}N(i),$$
$\textrm{ where } w_i- c_i \in \{1,2,3\}.$ 

\item Suppose $c_i < c_{i+1} < w_i $ and $c_{i+2} = c_{i+1}$, which implies that $\mathrm{dep}_i(c_{i+1}): 0 \rightarrow 3$. Thus, 
$$g_i^{\Delta} = \frac{1}{4}(w_i-c_i-1)N(i+1), \textrm{ where } w_i- c_i \in \{2,3\}.$$

\item Suppose $c_{i+1} < c_i < w_i $, which implies that there is no change in the depth $\mathrm{dep}_i(c_{i+1})$. Thus, 
$$g_i^{\Delta} = \frac{1}{4}(w_i-c_i)N(i+1), \textrm{ where } w_i- c_i \in \{1,2\}.$$
\end{enumerate}

We now study $g_{i-1}(c_{i+1}c_i c_{i-1} \rightarrow c_{i+1}w_i c_{i-1})$. Note that if $c_{i-1} \leq c_i$, then $g_{i-1}(c_{i+1}c_i c_{i-1} \rightarrow c_{i+1}w_i c_{i-1})=0$. Thus, a nonzero change in symbol contribution at index $i-1$ can occur, if any, due to changes in 
\begin{enumerate}
\item the depth of $c_i$ at $i-1$ if $c_i < c_{i-1} \leq w_i$, or
\item the depths of $c_i$ and $w_i$ at $i-1$ if $c_i < w_i < c_{i-1}$.
\end{enumerate}
 
In the first case, the depth of $c_i$ at $i-1$ is decreased from $2$ to $1$ (one $\delta_{i-1,2}$ becomes $0$, whereas $\delta_{i-1,1}$ becomes $1$) or from $3$ to $1$ (one $\delta_{i-1,1}$ becomes $1$),
and thus using (\ref{eq:explicit_exp_ell=2}) 
\begin{align*} g_{i-1}(c_{i+1}c_i c_{i-1} \rightarrow & \textrm{ } c_{i+1}w_i c_{i-1}) = \frac{3}{4}N(i-2) \textrm{ or}\\
g_{i-1}(c_{i+1}c_i c_{i-1} \rightarrow & \textrm{ } c_{i+1}w_i c_{i-1}) \\
& = \frac{3}{4}(N(i-1)+N(i-2)).
\end{align*}

In the second case, we have the following possibilities:
\begin{enumerate}
\item $\mathrm{dep}_{i-1}(c_i): 2 \rightarrow 1$ and $\mathrm{dep}_{i-1}(w_i): 1 \rightarrow 2$ or 
 \item $\mathrm{dep}_{i-1}(c_i): 2 \rightarrow 1$ and $\mathrm{dep}_{i-1}(w_i): 1 \rightarrow 3$ or 
 \item $\mathrm{dep}_{i-1}(c_i): 3 \rightarrow 1$ and $\mathrm{dep}_{i-1}(w_i): 1 \rightarrow 2$, and thus
\end{enumerate}
\begin{align*} & g_{i-1}(c_{i+1}c_i c_{i-1} \rightarrow c_{i+1}w_i c_{i-1}) = 0 \textrm{ or} \\
& g_{i-1}(c_{i+1}c_i c_{i-1} \rightarrow c_{i+1}w_i c_{i-1}) = -\frac{3}{4}N(i-1) \textrm{ or} \\
& g_{i-1}(c_{i+1}c_i c_{i-1} \rightarrow c_{i+1}w_i c_{i-1}) = \frac{3}{4}N(i-1).
\end{align*}

Finally, we study $g_{i-2}(c_i c_{i-1} c_{i-2} \rightarrow w_i c_{i-1} c_{i-2})$. Note that a nonzero change in symbol contribution at index $i-2$ occurs due to nonzero changes in the depth of $c_{i-1}$ at $i-2$. It changes either from $3$ to $2$ or from $2$ to $3$, and thus 
\begin{align*}
& g_{i-2}(c_i c_{i-1} c_{i-2} \rightarrow w_i c_{i-1} c_{i-2}) = \frac{3}{4}N(i-2) \textrm{ or} \\
& g_{i-2}(c_i c_{i-1} c_{i-2} \rightarrow w_i c_{i-1} c_{i-2}) = -\frac{3}{4}N(i-2).
\end{align*}
Moreover, $g_{i-2}(c_i c_{i-1} c_{i-2} \rightarrow w_i c_{i-1} c_{i-2})$ can be zero as well, in case $c_{i-1} \neq w_i$ and $c_{i-1} \neq c_i$ for instance. 

\begin{example}
For the case of $c_{i+2} c_{i+1} c_i c_{i-1} c_{i-2} = TTAGC$ and $w_i = G$, we have $\mathrm{dep}_i(c_{i+1} = T): 0 \rightarrow 3$, leading to $g_i(TTA \rightarrow TTG) = N(i+1)/4$. Moreover, $\mathrm{dep}_{i-1}(c_i = A): 2 \rightarrow 1$ and $\mathrm{dep}_{i-1}(w_i = G): 0 \rightarrow 0$, leading to $g_{i-1}(TAG \rightarrow TGG) = 3N(i-2)/4$. Finally, $\mathrm{dep}_{i-2}(c_{i-1} = G): 2 \rightarrow 3$, leading to $g_{i-2}(ACG \rightarrow GCG) = -3N(i-2)/4$. Hence, the index difference $g(\mathbf{w},2)-g(\mathbf{c})$ is $N(i+1)/4$.
\end{example}

In the next section, we analyze the interdependence of changes in symbol contributions in order to achieve notably lower storage overhead and increase rates. 

\vspace{-1em}
\subsection{Interdependence of Symbol Contributions for $\ell=2$} \label{sec:interdependence_for_ell=2}
In this section, we study the interdependence of the changes in symbol contributions for general $m$ to obtain a set of non-negative index errors, containing all possible non-negative index differences. Some of the contribution change scenarios at indices $i$, $i-1$, and $i-2$ discussed above are mutually exclusive. Thus, we will eliminate some matchings based on that the original codeword is free of runs of length $> 2$. This \textit{optimization} reduces storage overhead. See Remark~\ref{rem:correcting_words_2} for the consequence of this (concerning negative index differences).

\begin{remark} \label{formal_equality} In the following analysis, whenever we discuss the impossible matchings for $g_i^{\Delta}$, $g_{i-1}^{\Delta}$, and $ g_{i-2}^{\Delta}$, we mean as formal expressions, and do not mean as integer values. An integer value resulting from an eliminated matching can still be attained due to a non-eliminated matching.
\end{remark}

The set of positive index differences due to a substitution at locations $0$ or $1$ is $\{0,1,\dots, 12\} \setminus \{5, 6, 9, 10\}$, where only $\{0,1,2,3\}$ are possible if the substitution is at location $0$ and $\{5, 6, 9, 10\}$ are removed because of impossible matchings. We keep the assumption $i \geq 2$ until the end of the section. 
\par
\textbf{Case 1:} Note that $g_i^{\Delta}$ can take the expressions $\frac{1}{4}N(i+1)$ or $\frac{2}{4}N(i+1)$ only if $c_{i+1} \neq c_i$. In which case, $g_{i-1}^{\Delta}$ cannot take the expressions $\frac{3}{4}(N(i-1)+N(i-2))$ or $\frac{3}{4}N(i-1)$ because the initial depth of $c_{i}$ at $i-1$ cannot be $3$.

\textbf{Case 2:} Note that $g_i^{\Delta}$ can take the expression $\frac{3}{4}N(i+1)$ only if $c_i =A$ and $w_i = c_{i+1} = C$. In which case, the only possible expressions for $g_{i-1}^{\Delta}$ are $\frac{3}{4}N(i-2)$ or $0$, and $g_{i-2}^{\Delta}$ cannot take the expression $-\frac{3}{4}N(i-2)$ (since the final depth of $c_{i-1}$ at $i-2$ cannot be $3$).
 
Moreover, note that we cannot observe $g_{i-1}^{\Delta}=0$ and $g_{i-2}^{\Delta}=0$ simultaneously in this case. For $g_{i-1}^{\Delta}$ to take $0$, we must have $c_{i-1}=A$ since otherwise, $\mathrm{dep}_{i-1}(c_i=A):2\rightarrow 1$, whereas $\mathrm{dep}_{i-1}(w_i=C)$ is unaffected (and zero). For $g_{i-2}^{\Delta}$ to take $0$ as well, we must have $c_{i-2}=A$ since otherwise, $\mathrm{dep}_{i-2}(c_{i-1}=A)$ is affected. However, $c_i$, $c_{i-1}$, and $c_{i-2}$ cannot all be identical due to the run-length constraint on $\mathbf{c}$.

\textbf{Case 3:} Note that $g_i^{\Delta}$ can take the expression $\frac{3}{4}N(i)$ only if $c_i = c_{i+1}$ and $w_i$ is the next symbol after $c_i$ (with respect to the ordering $A < T < G < C$). In which case, $g_{i-1}^{\Delta}$ cannot take the expressions $\frac{3}{4}N(i-2)$ or $-\frac{3}{4}N(i-1)$ as the initial depth of $c_i$ at $i-1$ cannot be $2$. This also implies that for $g_{i-1}^{\Delta}$ to take $0$, we must have $c_{i-1} \leq c_i=A$, which is not possible due to the run-length constraint on $\mathbf{c}$. Moreover, $g_{i-2}^{\Delta}$ cannot take the expression $\frac{3}{4}N(i-2)$ since $c_{i+1}=c_i$ and $\ell = 2$ imply that $c_{i-1} \neq c_i$, which means the initial depth of $c_{i-1}$ at $i-2$ cannot be $3$.

\textbf{Case 4:} Since $g_i^{\Delta} = \frac{1}{4}N(i+1)+\frac{3}{4}N(i)$ only if $c_{i+1} < w_i$, $g_{i-1}^{\Delta}$ cannot take the expression $-\frac{3}{4}N(i-1)$ as the final depth of $w_i$ at $i-1$ cannot be $3$.

\textbf{Case 5:} Note that $g_i^{\Delta}$ can take the expression $\frac{2}{4}N(i+1)+\frac{3}{4}N(i)$ only if $c_i = c_{i+1}= A$ and $w_i =C$. In which case, $c_{i-1} \neq A$ and the only possible expression for $g_{i-1}^{\Delta}$ is $\frac{3}{4}(N(i-1)+N(i-2))$ as the initial depth of $c_i$ at $i-2$ cannot be $2$ and the final depth of $w_i= C$ at $i-1$ is $0$ (by definition). Moreover, $g_{i-2}^{\Delta}$ cannot be $-\frac{3}{4}N(i-2)$ because the final depth of $c_{i-1}$ at $i-2$ cannot be $3$.

\begin{definition} \label{def:indexset2} We denote the set 
\begin{align*} \big\{ \big[i; & \textrm{ } \theta_1 \cdot N(i+1)/4 + \theta_2 \cdot 3N(i)/4 \big] \textrm{ } | \textrm{ } 2 \leq i \leq m-1, \nonumber \\ & (\theta_1,\theta_2) \in \{(1,0),(2,0),(3,0),(0,1), (1,1),(2,1)\}\big\}
\end{align*}
of pairs by $\mathsf{F}_1$, the set
\begin{align*}  \big\{ \big[i-1; \textrm{ } \theta_1 \cdot 3N(i-1)/4  &+  \theta_2 \cdot 3N(i-2)/4 \big] \textrm{ } | \textrm{ } \\
2 \leq i \leq m-1 &,  \\
(\theta_1,\theta_2) \in \{(0,0), (0,1),(1,1), & (1,0),(-1,0)\}\big\}
\end{align*}
of pairs by $\mathsf{F}_2$, and the set
\begin{align*} \big\{ \big[i-2; & \textrm{ } \theta \cdot 3N(i-2)/4 \big] \textrm{ } | \textrm{ } 2 \leq i \leq m-1, \theta \in \{-1, 0,1\}\big\}
\end{align*}
of pairs by $\mathsf{F}_3$. Here, $i$ is the location of the substitution.

\noindent The set $\mathcal{E}_{\textrm{sup}}^+(m, 2)$ of non-negative index errors for $\ell=2$ is defined as 
\begin{align*}
\bigg \{ &x + y + z  \textrm{ } | \textrm{ } \big[ [i; x] , [i-1;y] , [i-2;z] \big] \in \mathsf{F}_1 \times \mathsf{F}_2 \times \mathsf{F}_3 \, \setminus \, \mathcal{I}  \\
& \textrm{ for all } 2 \leq i \leq m-1\bigg\} \, \cup \, \big(\{0,1,\dots, 12\} \, \setminus \, \{5,6,9,10\}\big),
\end{align*}
where $\mathcal{I} \subset \mathsf{F}_1 \times \mathsf{F}_2 \times \mathsf{F}_3$ is the set of impossible matchings.
\end{definition}

\section{Residue Decoding} \label{sec:residue_algorithm}
In this section, we introduce our EC D-LOCO codes and their index-correction (error-correction) algorithm, namely the residue decoding algorithm. For a positive integer $n$, $[n]$ denotes the set $\{0,1,\dots,n-1\}$. 

\begin{lemma} \label{lemma:zero}
Fix a positive integer $R$. Let $N_1, N_2$ be two non-negative numbers such that $N_1$ is divisible by $R$. Let ${\Phi}: \mathbb{Z} \rightarrow [R]$ be the residue map, i.e., ${\Phi}(k) = (k \textup{ } \mathrm{mod} \textup{ } R) \in [R]$. Then,
\[ {\Phi}(N_2) = \begin{cases} {\Phi}(N_2-N_1) , \hspace{0.8cm} \textrm{ if } N_2 \geq N_1, \\
R-{\Phi}(N_1-N_2),  \hspace{0.2cm} \textrm{ otherwise.}
\end{cases}
\]
\end{lemma}

\begin{IEEEproof} The proof follows from the definition of the $\mathrm{mod}$ operation. Note that ${\Phi}(N_2-N_1) = R-{\Phi}(N_1-N_2)$.
\end{IEEEproof}

When we apply Lemma~\ref{lemma:zero} in the context of this work, $N_1$ will be the index of the original codeword and $N_2$ will be the formal index of the received word. Recall that the \textit{residue} of an integer is this integer modulo $R$. The lemma then suggests that the index difference $N_2-N_1$ can be recovered from the residue of $N_2$ provided that $R$ is chosen (for fixed $m$ and $\ell$) to guarantee the injectivity of ${\Phi}$ on the set of signed index differences (all index differences are covered). Our coding scheme will incorporate a table, namely the residue table, of positive index errors and their residues. It is important to find the smallest working $R$ in order to maximize the finite-length rates and minimize the required storage for the table.

We outline an algorithm to find the \textit{redundancy metric} $R(m, \ell)$ (as a function of the code length $m$ and maximum allowed run-length $\ell$) which guarantee $100\%$ correction of single-substitution errors via the residue decoding algorithm. For brevity, we denote the redundancy metric by $R(m)$ (or by $R$) if the context clarifies $\ell$ (and also $m$).
\begin{enumerate} \label{algo:findR}
\item[1.] Set $R=2$. 
\item[2.] Check whether the residues of elements in the (extended) set of signed index errors 
\begin{align*} \mathcal{E}_{\textrm{sup}}(m, \ell) \triangleq \{-\mathsf{e} \textrm{ } | \textrm{ } \mathsf{e} &\in \mathcal{E}_{\textrm{sup}}^+(m,\ell)\} \cup \mathcal{E}_{\textrm{sup}}^+(m,\ell) \\
&\cup \{0, N_{\mathrm{D}}(m,\ell)-1\}, \textrm{ } \ell \in \{1,2\},
\end{align*} 
are pairwise distinct (see Remark~\ref{rem:about_R} regarding why we include $N_{\mathrm{D}}(m,\ell)-1$).

\item[3.] In the case of failure, increase $R$ by one and return to Step~$2$ until success occurs. Upon success, output $R$.
\end{enumerate} 

\begin{definition} \label{def:residual} For a codeword $\mathbf{c} \in \mathcal{D}_{m,\ell}$, $(g(\mathbf{c}) \textup{ } \mathrm{mod} \textup{ } R)$ is called its residual index. For a word $\mathbf{w} \in \{A,T,G,C\}^m$, the number $(g(\mathbf{w}, \ell) \textup{ } \mathrm{mod} \textup{ } R)$ is called its formal residual index.
\end{definition}

Next, we introduce EC D-LOCO codes, the primary contribution of this paper.
 
\begin{definition} \label{def:EC-D-LOCO} (EC D-LOCO code): The EC D-LOCO code, denoted by $\mathcal{D}^{\mathsf{Res}}_{m,\ell}$, with parameters $m \geq \ell \geq 1$, is the set
\[ \{\mathbf{c} \in \mathcal{D}_{m, \ell} \textrm{ } | \textrm{ } (g(\mathbf{c}) \textup{ } \mathrm{mod} \textup{ } R) = 0\}, \]
which is the set of all codewords in $\mathcal{D}_{m,\ell}$ with zero residual index. We refer to its elements as EC (D-LOCO) codewords. Its finite-length rate is 
\begin{align} \frac{\big \lfloor \log_2((N(m)-1)/R+1) \big \rfloor}{m+3}.
\end{align}
\end{definition}

\noindent The addition of $3$ in the denominator is due to bridging (to be discussed in the section below). See Table~\ref{tableR2} ($\ell = 2$) for sample redundancy metrics and rates at various lengths.\footnote{These values are based on a relaxed injectivity condition on ${\Phi}$, which will be discussed in Section~\ref{sec:storage-overhead}.}

\begin{table}[t!]
\caption{Adopted Redundancy Metric $R$ and The Corresponding Rates at Various Lengths for $\ell=2$}
\vspace{-0.1em}
\centering
\scalebox{1.1}
{
\begin{tabular}{|c|c|c|c|c|}
\hline
\multicolumn{3}{|c|}{\makecell{$\ell =2$}}  \\
\hline
$m$ & $R(m)$ & \textrm{Rate} \\
\hline
$17$ & $9766$ & $0.9500$ \\
\hline
$27$ & $22045$  & $1.2333$ \\
\hline
$33$ & $45418$ & $1.3333$ \\
\hline
$37$ & $49981$ & $1.3750$ \\
\hline
$47$ & $80993$ & $1.4800$ \\
\hline
$55$ & $114088$ & $1.5344$\\
\hline
$61$ & $137389$ & $1.5625$\\
\hline
\end{tabular}}
\label{tableR2}
\vspace{-0.5em}
\end{table}

\subsection{Bridging Scheme and Balancing} \label{sec:bridge_balance}
We will now describe a $3$-symbol bridging scheme that will play a crucial role in the residue decoding algorithm. Let $\mathbf{d}_1$ and $\mathbf{d}_2$ be two consecutive EC codewords in the DNA data stream, where $\mathbf{d}_1$ ends with the symbol $\Lambda_1$ and $\mathbf{d}_2$ starts with the symbol $\Lambda_2$. For clarity, the sequence order will be $\Lambda_1 \ \Lambda_4 \Lambda_3 \Lambda_5 \ \Lambda_2$ below. The middle bridging symbol $\Lambda_3$ will be the check-sum $\sum_{i=0}^{m-1} a_i$ of the previous codeword according to the correspondence in (\ref{eqn_int}) for $\ell=2,3$. For $\ell =1$, we set $\Lambda_3$ to $\sum_{i=1}^{m-1} a_i$ instead to handle the errors at the right-most symbol of the codeword. We call $\Lambda_3$ a local detection check-sum. 

Define \textit{disparity} as the difference between the number of $GC$ and $AT$ symbols. Complement symbols are said to have opposite disparity. For $\ell = 2,3$, the first bridging symbol $\Lambda_4$ will be set to 
\begin{itemize}
\item[-] the letter with the highest lexicographic index having opposite disparity to $\Lambda_1$ if the previous codeword is complemented due to balancing requirement, and to
\item[-] the letter with the lowest lexicographic index in having opposite disparity to $\Lambda_1$ otherwise.
\end{itemize} 

In order not to create a forbidden pattern, for $\ell = 1$, the first bridging symbol $\Lambda_4$ will instead be set to 
\begin{itemize}
\item[-] the letter with the highest lexicographic index in the set $\{A, T, G, C\} \setminus \{\Lambda_1, \Lambda_3\}$ if the previous codeword is complemented due to balancing requirement, and to
\item[-] the letter with the lowest lexicographic index in the set $\{A, T, G, C\} \setminus \{\Lambda_1, \Lambda_3\}$ otherwise. 
\end{itemize} 

The third bridging symbol $\Lambda_5$ will be set to the letter with the highest lexicographic index in the set $\{A, T\} \setminus \{\Lambda_2\}$ in case $\Lambda_3 \in \{G, C\}$ and in the set $\{G, C\} \setminus \{\Lambda_2\}$ in case $\Lambda_3 \in \{A, T\}$ also to maintain $GC$-content balance. 

Our original D-LOCO balancing technique is based on the idea that each D-LOCO codeword and its complement encode the same message. The codeword to use is that of opposite disparity to the stream disparity prior to it. See \cite[Section~IV]{D-LOCO} and \cite[Section~V]{D-LOCO} for more details about bridging and balancing, respectively.

\begin{remark} \label{rem:about_R} Observe that the sum of the indices of a codeword and its complement is always $N_{\mathrm{D}}(m,\ell)-1$. Therefore, the residual index of a codeword complement is always $((N_{\mathrm{D}}(m,\ell)-1) \textup{ } \mathrm{mod} \textup{ } R) \neq 0$. By including $N_{\mathrm{D}}(m,\ell)-1$ in the set $\mathcal{E}_{\textrm{sup}}(m, \ell)$, we ensure the following: 
\begin{itemize} 
\item[(i)] The complement of an EC codeword has a saved nonzero residual index, allowing us to distinguish an EC codeword from its complement (which is not itself an EC codeword by definition) in case the first bridging symbol cannot be trusted.

\item[(ii)] More crucially, error-free EC codeword complements can be distinguished from erroneous EC codewords (single substitution) by their residual indices in case the local detection check-sum cannot be trusted.
\end{itemize} 
Having (i) allows us to use the entire EC D-LOCO codebook (up to floor the cardinality) for distinct messages with EC D-LOCO balancing, i.e., precisely $2^{\lfloor \log_2((N(m)-1)/R+1) \rfloor}$-many EC codewords are at our disposal, whereas D-LOCO balancing requires to give up on ``half of the D-LOCO codebook'' for distinct messages.
\end{remark}

Based on that, a DNA strand consisting of $K$ EC D-LOCO codewords of odd length $m$ will have disparity in the range $[-m-1, m+1]$, achieving $GC$-content $40\%-60\%$ for $K \geq 5$.

\begin{definition} An EC D-LOCO codeword $\mathbf{d}$ in a 
DNA strand together with the three following bridging symbols is called the concatenable of $\mathbf{d}$ (or just a concatenable if $\mathbf{d}$ need not be emphasized).
\end{definition}

\vspace{-0.5em}
\subsection{Steps of the Residue Decoding Algorithm} 
The encoding of EC D-LOCO codes is as follows: \\
Let $R$ be the redundancy metric for fixed $m$ and $\ell \in \{1,2,3\}$. \\
\textbf{Input.} $\big \lfloor \log_2((N_{\mathrm{D}}(m,\ell)-1)/R+1) \big \rfloor$-bit binary data $\mathbf{b}$.  \\
\textbf{Output.} $\mathbf{d}^{\mathrm{c}}$ for some EC D-LOCO codeword $\mathbf{d} \in \mathcal{D}^{\mathsf{Res}}_{m,\ell}$.
\begin{enumerate}
\item[E1.] Multiply $\mathrm{decimal}(\mathbf{b})$ by $R$ and then encode it by the D-LOCO encoder $\mathrm{E}$ \cite[Algorithm~1]{D-LOCO} to obtain the codeword $\mathbf{d}=\mathrm{E}(\mathrm{binary}(\mathrm{decimal}(\mathbf{b})*R))$.
\item[E2.] Complement the codeword $\mathbf{d}$, if necessary, to control the $GC$-content based on EC D-LOCO balancing. We denote the possibly complemented codeword by $\mathbf{d}^{\mathrm{c}}$.
\end{enumerate}
Here, $\mathbf{d}^{\mathrm{c}}$ is either $\mathbf{d}$ itself or its complement $\overline{\mathbf{d}}$, and thus $(\mathbf{d}^{\mathrm{c}})^{\mathrm{c}}=\mathbf{d}$. Note that $\overline{\mathbf{d}}$ should satisfy $\overline{\mathbf{d}} = \mathrm{E}(\mathrm{binary}(N(m)-1-g(\mathbf{d})))$ and $g(\overline{\mathbf{d}}) \equiv N(m)-1 \textup{ } (\mathrm{mod} \textup{ } R)$.
\par
The residue decoding of EC D-LOCO codes is as follows: 
\\
\textbf{Assumption 1.} Suppose that at most one substitution error occurs at some symbol of the codeword $\mathbf{d}^{\mathrm{c}}$ or at the following three bridging symbols (i.e., in the concatenable of $\mathbf{d}^{\mathrm{c}}$). We denote the (possibly) erroneous word by $\mathbf{w}$.  \\
\textbf{Input.} A DNA sequence of length $m+3$ that differs at one location at most from the concatenable of $\mathbf{d}^{\mathrm{c}}$ for some EC codeword $\mathbf{d}$. Here, the received word $\mathbf{w}$ is retrieved as the segment of the $m$ symbols from the left where the following $3$ symbols are (possibly erroneous) bridging symbols. \\
\textbf{Output.} The index $g(\mathbf{d})$. 
\begin{enumerate}
\item[D0.] Compute the formal residual index $(g(\mathbf{w},\ell) \textup{ } \mathrm{mod} \textup{ } R)$. 
\item[D1.] There are two possible cases:
\begin{itemize} 
\item[Case 1.] The received word $\mathbf{w}$ is a D-LOCO codeword (i.e., free of runs of length $> \ell$) so that $g(\mathbf{w},\ell)= g(\mathbf{w})$:
\begin{itemize}
\item[(1.1)] If the residual index $(g(\mathbf{w}) \textup{ } \mathrm{mod} \textup{ } R)$ is zero, then $\mathbf{w}$ has no substitution error and the encoding codeword $\mathbf{d}$ was not complemented during encoding (Step~E2).  We recover $\mathbf{d}=\mathbf{w}$ and return the index $g(\mathbf{w})$.

\item[(1.2)]  If the residual index $(g(\mathbf{w}) \textup{ } \mathrm{mod} \textup{ } R)$ is equal to $((N(m)-1) \textup{ } \mathrm{mod} \textup{ }R)$, then $\mathbf{w}$ has no substitution error and the encoding codeword $\mathbf{d}$ was complemented during encoding (Step~E2). We recover $\mathbf{d}=\overline{\mathbf{w}}$ and return the index $g(\overline{\mathbf{w}})$ (which is $N(m)-1-g(\mathbf{w}))$. 

\item[(1.3)] If $(g(\mathbf{w}) \textup{ } \mathrm{mod} \textup{ } R) \neq 0$ and $(g(\mathbf{w}) \textup{ } \mathrm{mod} \textup{ } R) \neq ((N(m)-1) \textup{ } \mathrm{mod} \textup{ } R)$, then we conclude that $\mathbf{w}$ is erroneous (whereas the first bridging symbol is guaranteed to be true). We infer the correct information from the first bridging symbol about whether the previous codeword is complemented or not for $\ell \in \{2,3\}$.\footnote{For $\ell = 1$, if the local check-sum is satisfied, the right-most symbol of $\mathbf{w}$ must be the erroneous one. In which case, we may not infer whether $\mathbf{d}$ is complemented or not for balancing even if the first bridging symbol is guaranteed to be true. In this case, $\mathbf{d}$ is set to the unique EC codeword in the set of $\{w_{m-1}w_{m-2} \cdots w_1\Pi, \overline{w_{m-1}w_{m-2} \cdots w_1 \Pi} \,|\, \Pi \in \{A,T,G,C \} \setminus \{w_0, w_1\}\}$ of at most six codewords, Step~D2 is skipped, and we proceed with Step~D3 in Section~\ref{sec:storage-overhead}. See Remark~\ref{rem:ell1_subtlety} for the uniqueness of such $\mathbf{d}$.} Complement the received codeword $\mathbf{w}$, if necessary, based on the (trusted) first bridging symbol. We denote the resulting codeword by $\mathbf{w}^{\mathrm{c}}$ and proceed with Step~D2.
 \end{itemize}
 \item[Case 2.] The received word $\mathbf{w}$ is not a D-LOCO codeword. Here, we apply the correction step in Remark~\ref{rem:correcting_words_2} (for illustration, see Example~\ref{exp:correction}, Scenario 4 and Example~\ref{exp2:correction}). Proceed with Step~D3 in Section~\ref{sec:storage-overhead}.
\end{itemize}

\begin{figure*}[ht!]
\vspace{-1em}
\center
\includegraphics[trim={0.0in 0.0in 0.0in 0.0in}, width=5.5in]{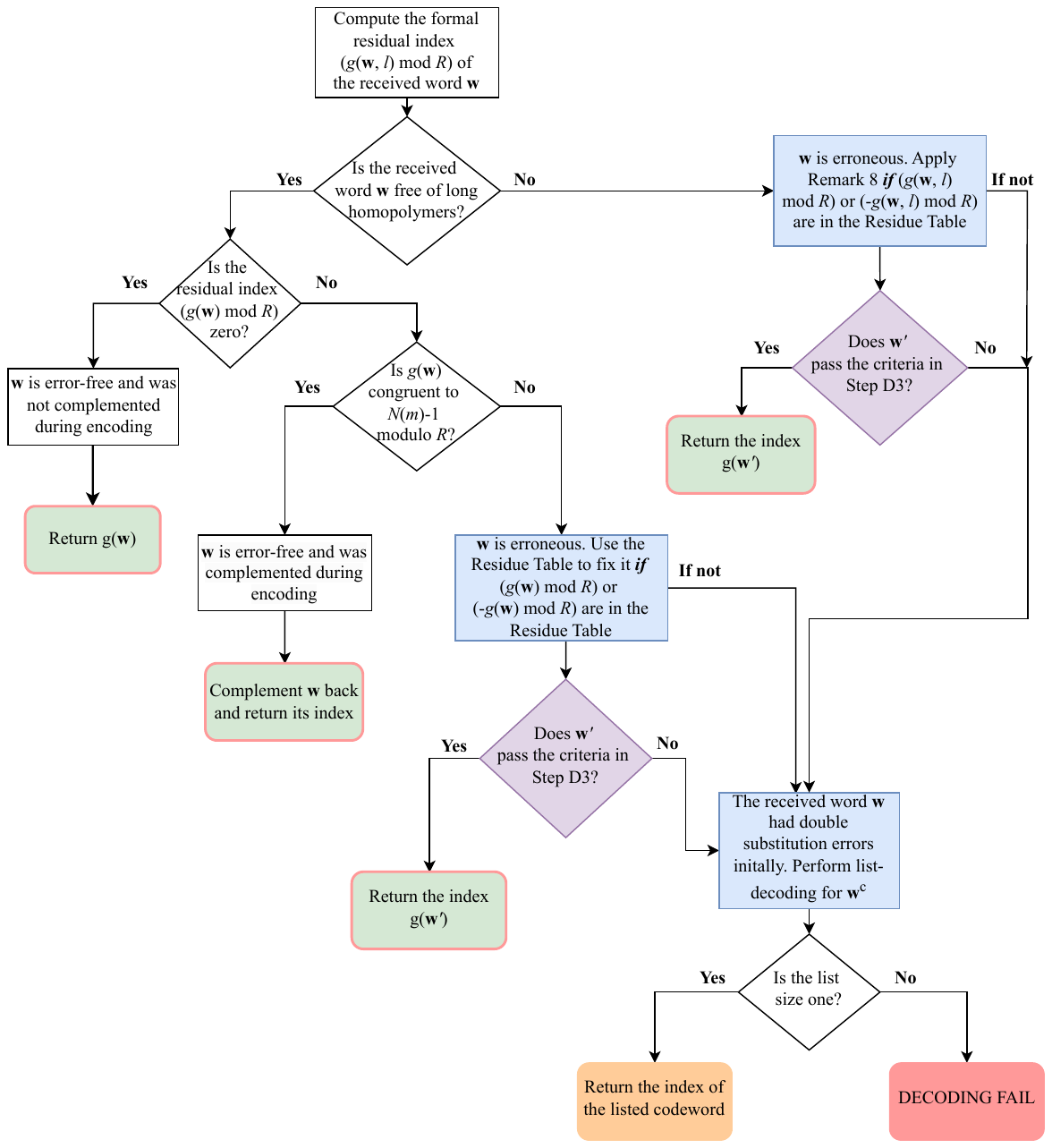}
\vspace{-14em}
\caption{A flowchart illustrating the steps of the residue decoding algorithm. Green boxes indicate a successful decoding under Assumption~1. The orange box indicates a successful decoding under Assumption~2 (see Section~\ref{sec:correcting_double_errors}). Red frames around green boxes as well as the red box (notably more probable than scenarios indicated by red frames) indicate failed decoding if a word under Assumption~2 (with two substitution errors) is inputted.}
\label{fig:algo}
\vspace{-0.5em}
\end{figure*}

\item[D2.] Compute the residual index $(g(\mathbf{w}^{\mathrm{c}}) \textup{ } \mathrm{mod} \textup{ } R)$. Search the residue table for both $(g(\mathbf{w}^{\mathrm{c}}) \textup{ } \mathrm{mod} \textup{ } R)$ and $(-g(\mathbf{w}^{\mathrm{c}}) \textup{ } \mathrm{mod} \textup{ } R)$. If $(g(\mathbf{w}^{\mathrm{c}}) \textup{ } \mathrm{mod} \textup{ } R)$ is found, decrease $g(\mathbf{w}^{\mathrm{c}})$ by the (unique) positive index error associated with the residue $(g(\mathbf{w}^{\mathrm{c}}) \textup{ } \mathrm{mod} \textup{ } R)$. If $(-g(\mathbf{w}^{\mathrm{c}}) \textup{ } \mathrm{mod} \textup{ } R)$ is found, increase $g(\mathbf{w}^{\mathrm{c}})$ by the (unique) positive index error associated with the residue $(-g(\mathbf{w}^{\mathrm{c}}) \textup{ } \mathrm{mod} \textup{ } R)$. This will result in the correct index $g(\mathbf{d})$ under Assumption~1.
\end{enumerate} 

\noindent We can then find the binary data $\mathrm{binary}(g(\mathbf{d})/R)$ of length $\big \lfloor \log_2((N(m)-1)/R+1) \big \rfloor$ for the corrected index. However, we do not yet specify this as a final step since the residue decoding algorithm will be later modified for optimal storage (see Step~OD2 in Section~\ref{sec:storage-overhead}) and extended by additional steps to handle double substitutions (see Step~D3 in  Section~\ref{sec:storage-overhead} and the steps in Section~\ref{sec:correcting_double_errors}). Fig.~\ref{fig:algo} depicts a flowchart of the complete algorithm to offer insight into our approach.

We now give various scenarios regarding Step~D2 and Case~2 in Step~D1, demonstrating the conceptual framework behind our correction algorithm. 

\begin{example} \label{exp:correction} For $m=6, \, \ell= 1$, we find $R = 127$ by the algorithm outlined in Section~\ref{sec:residue_algorithm}. Note that  $N(6)=972$ and $\lfloor \log_2((N(6)-1)/R + 1) \rfloor = \lfloor \log_2(971/127 + 1) \rfloor = 3$. Thus, we can encode 3-bit binary data using the following EC codewords in $\mathcal{D}^{\mathsf{Res}}_{6,1}$:
\[ ATATAT, \, AGTCAG, \, TATGAC, \, \cdots, \,  CTCGCT\] with respective indices $0, \, 127, \, 254, \, \cdots, \, 889$. No codeword need to be removed due to flooring in this case. Let $\mathbf{b}=001$ be the binary data, then $N_1=127*\mathrm{decimal}(001)=127$ and $AGTCAG$ encodes the data. By Definition~\ref{def:indexset1}, we have
\begin{align*} \mathcal{E}_{\textrm{sup}}^+(6,1) &= \{0, 1 , 2 , 3 , 4 , 6 , 7 , 9 , 10 , 12 , 18 , 21 , 27 , 30 , 36, 54,  \\
63 , & 81 , 90 , 108 , 162 , 189 , 243 , 270, 324, 486, 567, 810\}, 
\end{align*}
and the set of residual indices, denoted by ${\Phi}(\mathcal{E}_{\textrm{sup}}^+(6,1))$, is 
\begin{align*}  
{\Phi}(\mathcal{E}_{\textrm{sup}}^+(6,1)) &= \{0, 1 , 2 , 3 , 4 ,  6 , 7, 9 , 10 , 12 , 16, 18 , 21  , 27 , 30 , \\
35, & 36, 48, 54 , 59, 62, 63, 70, 81 , 90, 105, 108, 116\}. 
\end{align*}
 
\noindent In the following scenarios, we write erroneous symbols in bold.
\begin{enumerate}
\item
Suppose first that we receive $\mathbf{T}GTCAG$. Note that $g(\mathbf{T}GTCAG)= 370$, and thus the residual index is ${\Phi}(370)=116$. There is a single index error in $\mathcal{E}_{\textrm{sup}}^+(6,1)$, namely 243, congruent to ${\Phi}(370)$. Hence, the index of the encoding codeword is $370-243=127$.
\item
Suppose now that we receive $AG\mathbf{A}CAG$. Note that $g(AG\mathbf{A}CAG)= 100$, and thus the residual index ${\Phi}(100)=100$. There is no index error in $\mathcal{E}_{\textrm{sup}}^+(6,1)$ congruent to ${\Phi}(100)$. However, there is a (single) index difference, namely $27$, congruent to $-{\Phi}(100)$. This is due to the fact that we only store the positive index errors in the residue table.
Hence, we conclude that the index has decreased upon storage and the corrected index is $100+27=127$.
\item
Suppose that we receive the non-codeword $AG\mathbf{C}CAG$. It can be detected from the local detection check-sum ($=G$) that the index got increased upon storage. We have $g(AG\mathbf{C}CAG, 1)= 163$, and thus the formal residual index is ${\Phi}(163)=36$. The index difference $36 \in \mathcal{E}_{\textrm{sup}}^+(6,1)$ is congruent to ${\Phi}(163)$. Hence, the index of the encoding codeword is $163-36=127$.
\item
Suppose now that $TATGAC$ encodes a different message and we receive $TA\mathbf{A}GAC$, which is not a codeword. We have $g(TA\mathbf{A}GAC, 1)=254$ and the formal residual index is zero. Based on this, we conclude that the encoding codeword is the EC codeword of index $254$. We can of course just proceed with the index without obtaining the codeword in this case. 
\end{enumerate}
\end{example}

\begin{remark} \label{rem:formal_ham1} Formal indices do not respect the lexicographic ordering of words. However, in case we compare a codeword $\mathbf{c}$ and a word $\mathbf{w}$ with $\mathbf{w} > \mathbf{c}$ that are of Hamming distance $1$, there is a consistent observation that facilitates error correction. For instance, for $\ell=2$ and fixed $0 \leq i \leq m$, the contribution at $i$, $g_{i}$ for brevity, can increase by at least $3N(i)/4$, whereas $g_{i-1}$ and $g_{i-2}$ can decrease by at most $3N(i-1)/4$ and $3N(i-2)/4$, respectively. This implies that for $\mathbf{c} \in \mathcal{D}_{m,\ell}$ and a word $\mathbf{w}$ of Hamming distance $1$ to $\mathbf{c}$, if $\mathbf{w} > \mathbf{c}$, then $g(\mathbf{w}, \ell) \geq g(\mathbf{c})$. Otherwise, we have $g(\overline{\mathbf{w}}, \ell) \geq g(\overline{\mathbf{c}})$.
\end{remark}

\begin{remark} \label{rem:correcting_words_2} Upon storage, possible erroneous changes are: 
\begin{itemize}
\item[(i)] Codeword to codeword of higher index  
\item[(ii)] Codeword to non-codeword of higher or equal index 
\item[(iii)] Codeword to codeword of lower index
\item[(iv)] Codeword to non-codeword of lower or equal index.
\end{itemize}
A non-codeword here means a word that is not in $\mathcal{D}_{m,\ell}$. Once an exhaustive set of index errors for (i) is attained, then it suffices to take the negatives of its elements to cover (iii). However, the values for the last case (iv) cannot be covered by taking the negatives of the values from the first two cases as the optimization process in Section~\ref{sec:interdependence_for_ell=2} for the set of non-negative index errors is done under the fact that the encoding word is a codeword. We can deal with this subtlety without enlarging the size of the residue table and while maintaining guaranteed correction of single substitutions as follows. 
\\
Suppose we receive a non-codeword $\mathbf{w}$. If needed, we complement $\mathbf{w}$ based on the (trusted) first bridging symbol, obtaining an erroneous EC codeword whose erroneous symbol is one of the identical 
symbols in a run of length $> \ell$. The local detection check-sum will then tell us if the erroneous symbol needs be corrected to a symbol of lower or higher index. The former case can be handled as Step~D2. In the latter case, we complement the erroneous EC codeword as a correction trick, and find the formal index $N_2$ of its complement. We search the table for the residue $((N_2 - (N(m)-1)) \textup{ } \mathrm{mod} \textup{ } R)$ to neutralize the effect of the complementing trick. Remark~\ref{rem:formal_ham1} proves that this (modified) residue, if not $0$, will be found in the table since the index change from the codeword complement to the erroneous codeword complement is non-negative.
This will tell us how to correct the formal index $N_2$. After correcting $N_2$ to $N_2'$, we return the number $N(m)-1-N_2'$, where $N_2'$ is the index of the complement of the encoding EC codeword. 

We provide the pseudo-code for the case (iv) in Algorithm~\ref{dec_noncodeword} for $\ell \in \{2,3\}$, where the DNA symbols are overloaded with the respective integers equivalents.
\end{remark}

\begin{example} \label{exp2:correction} Under the same set-up of Example~\ref{exp:correction}, suppose~that $TATGAC$ encodes the message and we receive the word $TAT\mathbf{A}AC$, where one of $w_1$ or $w_2$ is erroneous.  The local detection check-sum ($=A$) indicates that, upon storage, one of $c_1$ or $c_2$ is replaced by a symbol ($=A$) of lower index. First, note that $g(TAT\mathbf{A}AC, 1)=245$, and ${\Phi}(245)=118$ is not in ${\Phi}(\mathcal{E}_{\textrm{sup}}^+(6,1))$. Following the steps outlined in Remark~\ref{rem:correcting_words_2}, we first complement the word $TAT\mathbf{A}AC$. For its complement, we have $g(GCG\mathbf{C}CA, 1)=729$ and ${\Phi}(729)=94$. The index difference $12 \in \mathcal{E}_{\textrm{sup}}^+(6,1))$ is congruent to $94-(N(6)-1)$. Hence, the index $729$ is corrected to $729-12=717$ (corresponding to the codeword $GCGTCA$). Finally, we recover the index of the encoding EC codeword as $(N(6)-1)-717=254$. 
\end{example}

\begin{algorithm}[t!]
{
\caption{Decoding EC D-LOCO Codes}
\begin{algorithmic}[1]
\State \textbf{Input:} A DNA sequence $w_{m-1}w_{m-2} \dots  w_{-2} w_{-3}$ of length $m+3$ that differs at one location at most from the concatenable of $\mathbf{d}^{\mathrm{c}}$ for some EC codeword $\mathbf{d}$. Here, the received word $\mathbf{w}$ is $w_{m-1}w_{m-2} \dots  w_{0}$ where the following $3$ symbols $w_{-1}, w_{-2} ,w_{-3}$ are (possibly erroneous) bridging symbols.
\State Compute the formal index $g(\mathbf{w},\ell)$ and the residual index $(g(\mathbf{w},\ell) \textup{ } \mathrm{mod} \textup{ } R)$ in a parallel manner (see Remark~\ref{rem:laten1}).
\State \textbf{if} $w$ is a D-LOCO codeword \textbf{then}
\State \hspace{2ex} $\cdots$
\State \textbf{else}
\State \hspace{2ex} Search for the longest run in $\mathbf{w}$, and set the repeating symbol to $\Lambda$.
\State \hspace{2ex} Find the symbol $\Lambda'$ so that $\sum_{i=0}^{m-1} w_i - \Lambda + \Lambda' \equiv w_{-2} \textup{ } (\textrm{mod} \textup{ } 4).$
\State \hspace{2ex} \textbf{if} $\Lambda > \Lambda'$ \textbf{then}
\State \hspace{4ex} $\cdots$
\State \hspace{2ex} \textbf{else}
\State \hspace{4ex} Compute the formal index $g(\overline{\mathbf{w}},\ell)$ and the residual index $(g(\overline{\mathbf{w}},\ell) \textup{ } \mathrm{mod} \textup{ } R)$ in a parallel manner.
\State \hspace{4ex} Search the residue table for $(g(\overline{\mathbf{w}},\ell) - (N(m)-1)) \textup{ } \mathrm{mod} \textup{ } R)$ and set $N'$ to the corresponding index error.
\State \hspace{4ex} \textbf{if} $w_{-1}=A \lor  w_{-1}=G$ \textbf{then} 
\State \hspace{6ex} Output $N(m)-1 -(g(\overline{\mathbf{w}},\ell)-N')$. 
\State \hspace{4ex} \textbf{else} 
\State \hspace{6ex} Output $g(\overline{\mathbf{w}},\ell)-N'$. 
\State \hspace{4ex} \textbf{end if} 
\State \hspace{2ex} \textbf{end if} 
\State \textbf{end if} 
\end{algorithmic}
\label{dec_noncodeword}
}
\end{algorithm}

\begin{remark} \label{rem:hamming} The Hamming distance between two EC D-LOCO codewords $\mathbf{d_1}, \mathbf{d_2} \in \mathcal{D}^{\mathsf{Res}}_{m,\ell}$ is guaranteed to be $> 2$ if
\begin{itemize}
\item[a)] there is a D-LOCO codeword $\mathbf{w}$ that is of Hamming distance $1$ to both $\mathbf{d_1}$ and $\mathbf{d_2}$, or
\item[b)] there is a non-codeword $\mathbf{w}$ that is of Hamming distance $1$ to both $\mathbf{d_1}$ and $\mathbf{d_2}$, and whose formal index is greater than or equal to both $g(\mathbf{d_1})$ and $g(\mathbf{d_2})$ or smaller than or equal to both $g(\mathbf{d_1})$ and $g(\mathbf{d_2})$. 
\end{itemize}
These can be shown via the run-length constraint and the choice of $R$. Unless otherwise proved, it is possible to have two EC codewords $\mathbf{d_1}$ and $\mathbf{d_2}$ of Hamming distance $2$, where the above conditions fail for any word $\mathbf{w}$ in the middle. However, this does not result in a decoding failure (under Assumption~1) in case such (a non-codeword) $\mathbf{w}$ is received since the trusted local detection check-sum will specify whether the erroneous symbol will be corrected to a higher or to a lower symbol, resolving the ambiguity.
\end{remark}

Next, we comment on the residue decoding latency.

\begin{remark} \label{rem:laten1} In Step~D0, the formal index and its residue can be computed in a parallel manner provided that the residues of cardinalities are stored. This is possible since the index computation involves $m-1$ summations so that one can realize the formal residual index as an additive combination of residues of cardinalities by additivity of the $\mathrm{mod}$ operation. 
\end{remark}

\section{Storage Overhead and Capacity Approachability} \label{sec:storage-overhead}

\subsection{Achieving Optimal Storage for Adopted Redundancy} 
In this section, we follow a standard method to decrease storage overhead and illustrate how to update Steps~D1 and ~D2 in Section~\ref{sec:residue_algorithm} accordingly. By studying the interdependence of changes in symbol contributions, we are able to shrink the height of the residue table. We will now shrink its width. 

Note that in the case of a single-substitution error, it suffices to locate the erroneous symbol to fix it since we have the check-sum (to be trusted) to tell what the right symbol must be. In order words, the index error is not necessary for index correction provided that it corresponds to a unique location of occurrence for the relevant residual index. In the residue table, for such index errors, we simply store the location of occurrence of errors in place of index errors. This also means we only need to distinguish index errors corresponding to distinct locations, relaxing the injectivity condition of ${\Phi}$ on $\mathcal{E}_{\textrm{sup}}(m, \ell)$. Experiments show that all index errors except for two values come from a substitution at a unique location for $\ell =2$ and odd $m \in [5, 61]$. The two exceptions in positive index errors are $3$ and $12$. For instance, 
\begin{align*} 
g(c_m\cdots GTC\mathbf{C}) & -g(c_m\cdots GTC\mathbf{A})=3, \\
g(c_m \cdots GT\mathbf{C}C) & -g(c_m \cdots GT\mathbf{G}C)=3, \\
g(c_{m-1} \cdots C\mathbf{C}A) & -g(c_{m-1} \cdots C\mathbf{A}A)=12 \textrm{, and}\\
g(c_{m-1}\cdots CT\mathbf{G}AA) & -g(c_{m-1}\cdots CT\mathbf{T}AA)=12.
\end{align*}

\vspace{-1em}
\begin{remark} \label{rem:ell1_subtlety} For $\ell = 1$, suppose that the received word is erroneous but the local check-sum is satisfied. This means that the erroneous symbol must be at the right-most location. In which case, $\mathbf{d}$ lies in the set $\{w_{m-1}w_{m-2} \cdots w_1\Pi, \overline{w_{m-1}w_{m-2} \cdots w_0 \Pi}\,|\, \Pi \in \{A,T,G,C \}\setminus \{w_0, w_1\}\}$ of at most six codewords. Only one of these codeword is an EC codeword since the residues of elements in $\{\pm, 3, \pm 2, \pm 1, 0, N(m,1)-1\} \subseteq \mathcal{E}_{\textrm{sup}}(m, 1)$ are pairwise distinct.
\end{remark}

\begin{enumerate}
\item[OD2.] Search the residue table for the residual index $(g(\mathbf{w}^{\mathrm{c}}) \textup{ } \mathrm{mod} \textup{ } R)$ as well as $(-g(\mathbf{w}^{\mathrm{c}}) \textup{ } \mathrm{mod} \textup{ } R)$. If $(g(\mathbf{w}^{\mathrm{c}}) \textup{ } \mathrm{mod} \textup{ } R)$ is found, then
\begin{itemize}
\item[-] if it is unlocatable, apply Step~D2 as in Section~\ref{sec:residue_algorithm} and encode the modified index (say, resulting in the word $\mathbf{w}'$), and

\item[-] if it is locatable, update the symbol of $\mathbf{w}^{\mathrm{c}}$ at the location corresponding to the residual index $(g(\mathbf{w}^{\mathrm{c}}) \textup{ } \mathrm{mod} \textup{ } R)$ in such a way that the local detection check-sum is satisfied. Denote the outcome by $\mathbf{w}'$ and proceed with Step~D3. The same is done for $(-g(\mathbf{w}^{\mathrm{c}}) \textup{ } \mathrm{mod} \textup{ } R)$. If neither is found (under the possibility of two or more substitution errors), proceed with Step~D4 in Section~\ref{sec:correcting_double_errors}. 
\end{itemize}
Case~2 in Step~D1 and the associated Remark~\ref{rem:correcting_words_2} can be adjusted in a similar manner. We call Step~D1 by Step~OD1 henceforth.

\item[D3.]  Check if the resulting word $\mathbf{w}'$ is a codeword with check-sum equal to that of $\mathbf{d}$, has Hamming distance $1$ to $\mathbf{w}$, and its residual index $(g(\mathbf{w}') \textup{ } \mathrm{mod} \textup{ } R)$ is zero. If so, which will be the case under Assumption~1, return $g(\mathbf{w}')$. If not, proceed with Step~D4 in Section~\ref{sec:correcting_double_errors}.
\end{enumerate}

\begin{remark} \label{rem:laten2}
In Step~OD1 and Step~OD2, once the location $i$ of substitution is determined from the residue table, one can simply subtract the $i^{\textrm{th}}, (i-1)^{\textrm{th}}, \dots, (i-\ell)^{\textrm{th}}$ symbol contributions of the erroneous codeword from the erroneous index and add those of the corrected (code)word. Therefore, the D-LOCO decoder runs only one time of $O(m)$ for the computation of (possibly erroneous) index, while the index correction requires only $O(\ell)$ steps, ensuring limited latency.
\end{remark}

For the complete decoding algorithm, we store 
\begin{enumerate}
\item[(1)] the cardinalities $\{3N(i)/4 \textrm{ } | \textrm{ } 0 \leq i \leq m-1\}$,
\item[(2)] the residue table, and
\item[(3)] the residues of cardinalities along with ${\Phi}(N(m)-1)$.
\end{enumerate}
(1) is required for index computation and (3) allows us to compute the residual index without computing the index itself. The required storage for (2) dominates over the other two. For $\ell = 2$, the number of positive index errors is around $20(m-2)$. For instance, the storage for (2) when $m=37$ is approximately
\begin{align*}
[20 \times (m-2)  &\times \big(\lfloor \log_2(m) \rfloor + 1 + \lfloor \log_2(R) \rfloor + 1 \big)\big] \\
= [20 \times 35 & \times (6 + 16)] \textrm{ } \textrm{bits} = 1.925 \textrm{ } \textrm{kB},
\end{align*}
and the total storage is about $2.17 \textrm{ } \textrm{kB}$ (see Table~\ref{tableR2} for $R(37)$).

\begin{table*}[ht]
\vspace{-1em}
\caption{A Comparison Between the Coding Scheme in \cite{nguyen_etal} and Ours Using EC D-LOCO Codes. $^{\textup{\tiny${\bigstar}$}}$ Indicates a Predicted Value Based on Section~\ref{sec:model}. Our Segment Length Is $m+3$}
\center
\scalebox{1.05}
{
\begin{tabular}{|c|c|c|c|c|c|} 
\hline
\multicolumn{6}{|c|}{\makecell{\textbf{Nguyen et al.'s Work} \cite{nguyen_etal}}} \\
\hline
\textrm{Strand (segment) length} & \textrm{$GC$-content} & \textrm{Max. run-length } & \textrm{Correction type per strand} & \textrm{Rate} & \textrm{Storage} \\
 \hline
$80 \textrm{ } (80)$ & $40\%-60\%$ & $3$ & \textrm{Single edit} & $1.3750$ & $\approx 0.50$ kB \\
 \hline
$146 \leq * \leq 148 \textrm{ } (*)$ & $40\%-60\%$ & $3$ & \textrm{Single edit} & $1.5946 \leq * \leq 1.6164$ & $\approx 2$ kB \\
 \hline 
\multicolumn{6}{|c|}{\makecell{\textbf{Our design using $\mathcal{D}^{\mathsf{Res}}_{m,2}$}}}  \\ 
\hline
$200 \textrm{ } (40)$ & $40\%-60\%$ & $2$ & \textrm{$5$ segmented subs. errors} & $1.3750$ & $\approx 2.17$ kB \\
 \hline
 $320 \textrm{ } (64)$ & $40\%-60\%$ & $2$ & \textrm{$5$ segmented subs. errors} & $1.5625$ & $\approx 4.13$ kB \\ 
 \hline
$400 \textrm{ } (80)$ & $40\%-60\%$ & $2$ & \textrm{$5$ segmented subs. errors}& $1.6250^{\textup{\tiny${\bigstar}$}}$ & $\approx 5.59^{\textup{\tiny${\bigstar}$}}$ kB \\ 
\hline
\end{tabular}}
\label{table:cmp}
\vspace{-1em}
\end{table*}

\begin{figure}[ht]
\vspace{-1em}
\center
\includegraphics[trim={0.0in 0.0in 0.0in 0.0in}, width=3.5in]{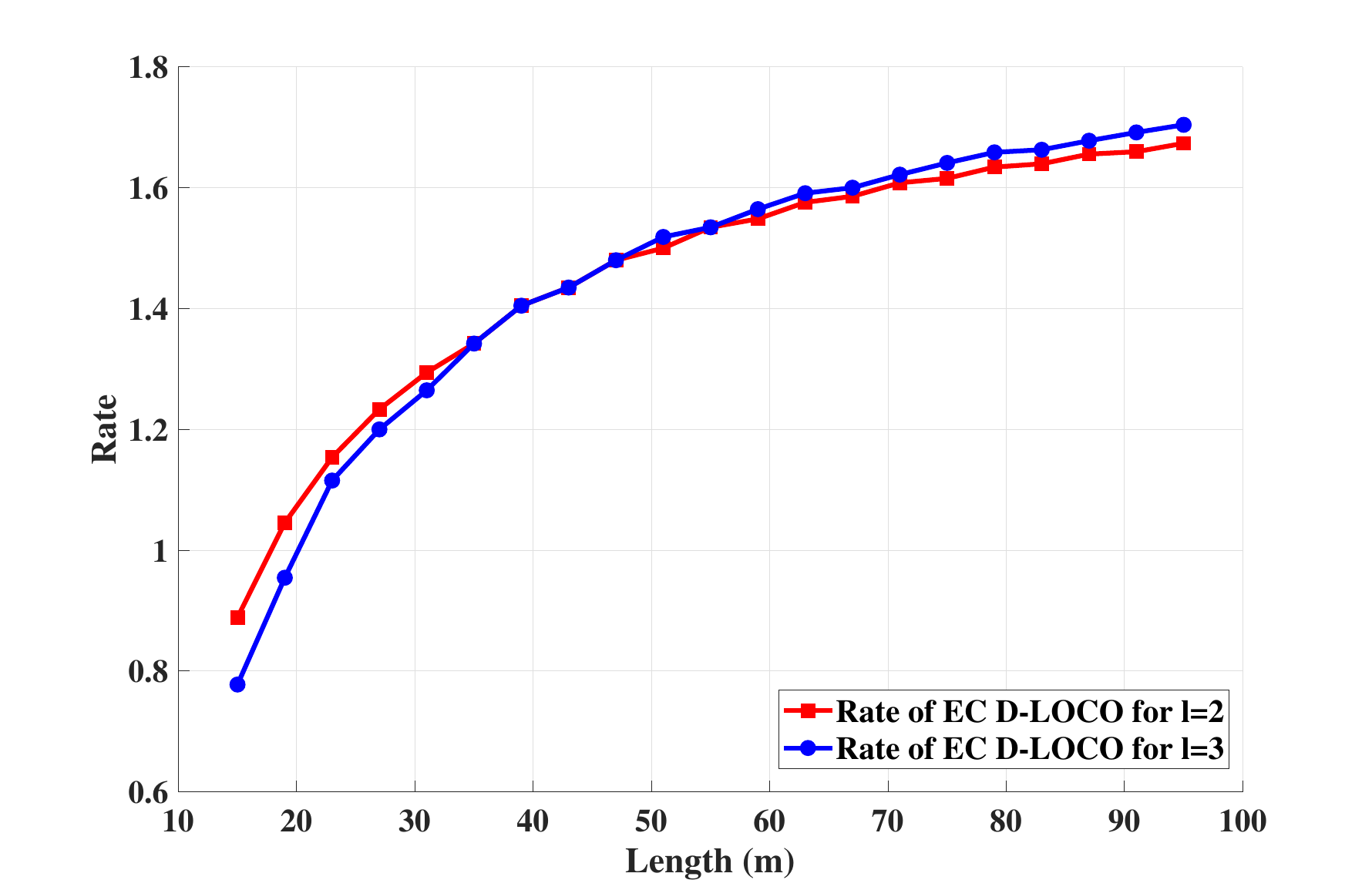}
\vspace{-1.5em}
\caption{Projected rates of EC D-LOCO codes $\mathcal{D}^{\mathsf{Res}}_{m,2}$ and $\mathcal{D}^{\mathsf{Res}}_{m,3}$ based on the quadratic modeling of $R(m,2)$ and $R(m,3)$. The rates on the red (blue) curve are actual values for $m \in [15,59]$ ($m \in [15,27]$).
}
\label{fig:cmp}
\vspace{-0.5em}
\end{figure}

\subsection{Modeling the redundancy metric} \label{sec:model}
In our prior work \cite{D-LOCO}, we showed that D-LOCO codes with $3$-symbol bridging possess error-detection property for single-substitution errors. Moreover, they are capacity-achieving as they have only $1$-bit penalty due to balancing. The capacity $C^{(2)}$, which is the base-$2$ logarithm of the largest real positive eigenvalue of the adjacency matrix of the finite-state transition diagram (representing the infinitude of a sequence forbidding runs of length $> 2$), is $1.9226$~\cite{immink_cai}. 

For EC D-LOCO codes $\mathcal{D}^{\mathsf{Res}}_{m,2}$, we performed curve-fitting for the data points $(m,R(m))$, where $m \in [5,61]$ is odd. The quadratic fitting polynomial turned out to be 
$35.43 \, m^2 + 110.42 \, m -2898.53$
whose root mean squared error (RMSE) normalized by the mean of data is $0.0945$. Based on this model, the asymptotic rate of $\mathcal{D}^{\mathsf{Res}}_{m,2}$ is
\begin{align}  \label{eq:capacity}
\lim_{m \rightarrow \infty} & \frac{\big \lfloor \log_2((N_{\mathrm{D}}(m,2)-1)/R(m)+1) \big \rfloor}{m+3} \nonumber \\
=&  \lim_{m \rightarrow \infty} \frac{ \log_2(N_{\mathrm{D}}(m,2)/R(m)) }{m} \nonumber\\
=&  \lim_{m \rightarrow \infty} \frac{ \log_2(N_{\mathrm{D}}(m,2))-\log_2(R(m)) }{m} \nonumber \\
=&  \lim_{m \rightarrow \infty} \frac{ \log_2(N_{\mathrm{D}}(m,2)) -2\log_2 m }{m} \nonumber \\
=& \lim_{m \rightarrow \infty} \frac{ \log_2 (N_{\mathrm{D}}(m,2) )}{m}, 
\end{align}
which is equal to $C^{(2)}$. Hence, $\mathcal{D}^{\mathsf{Res}}_{m,2}$ is expected to be capacity-approaching. 

\begin{remark} \label{rem:redundancy}
We call the redundancy $\log_2 (R(m, \ell))$ due to the adopted redundancy metric $R(m, \ell)$ the subcode or allocation redundancy. We achieve error correction for D-LOCO codes, not by adding redundant symbols but, by allocating bits from the binary data that a corresponding D-LOCO codeword can encode (without the correction property). As a result, we have a reduction by $\log_2 (R(m, \ell))$ in the numerator of the rate of EC D-LOCO codes in (\ref{eq:capacity}) above.
\end{remark}

We then summarize our findings for the case $\ell=3$. In this case, four symbol contributions are subject to change after a substitution at a location $i\geq 3$. The number of different values for $g_i^{\Delta}$, $g_{i-1}^{\Delta}$, $g_{i-2}^{\Delta}$, and $g_{i-3}^{\Delta}$ are $9$, $8$, $5$, and $3$, respectively. Due to the higher cardinality ($\approx 52(m-2)$ after optimization) of the set of index errors, the rates for $\mathcal{D}^{\mathsf{Res}}_{m,3}$ start lower than those of $\mathcal{D}^{\mathsf{Res}}_{m,2}$. Based on the quadratic (polynomial) modeling of the data $(m, R(m,3))$, $m \in [5,27]$, we find that the projected rates of $\mathcal{D}^{\mathsf{Res}}_{m,3}$ catch the actual rates of $\mathcal{D}^{\mathsf{Res}}_{m,2}$ at $m=35$, and  steadily surpass them after $m=55$ (see Fig.~\ref{fig:cmp}). 

\begin{remark}
For low lengths $m$, the set of index differences in Definition~\ref{def:set_difference} yields a notable storage gain for $\ell=2, 3$. This suggest that it is important to tackle this set in order to better analyze the rate-storage trade-off that $\ell \in \{2,3\}$ offers as well as achieving maximum EC D-LOCO rates.  
\end{remark}

\begin{remark} \label{rem:reconfigure} 
Reconfigurability of EC D-LOCO codes can be achieved at a storage and rate cost, for instance, by switching the run-length from $\ell$ to $\ell'$, where $\ell' < \ell$. Such reconfiguration is achieved simply by changing the inputs of the adders, i.e., the cardinalities, from $\{N_{\mathrm{D}}(i+1-\ell, \ell), \dots, N_{\mathrm{D}}(i, \ell)\}$ to $\{N_{\mathrm{D}}(i+1-\ell', \ell'), \dots, N_{\mathrm{D}}(i, \ell')\}$ through multiplexers to find the symbol contribution $g_i(c_i)$ and storing a second residue table with the redundancy metric $R(m,\ell')$ for $\mathcal{D}^{\mathsf{Res}}_{m,\ell'}$. This is an essential property in order to address device aging due to which homopolymers of shorter lengths become detrimental. 
\end{remark}

\section{Comparisons and Double Substitution Errors}\label{sec:comp_doub}
\subsection{Three Literature Comparisons} \label{sec:compare}
We compare the coding scheme of Nguyen et al.'s \cite{nguyen_etal} and our EC D-LOCO scheme using $\mathcal{D}^{\mathsf{Res}}_{m,2}$ for various $m$ in Table~\ref{table:cmp}. 
See Table~\ref{tableR2} and Fig.~\ref{fig:cmp} for our rates. 
\par
For instance, to design a strand of length $80$, they offer single-edit correction while achieving a rate of $1.3750$ with very low required storage. We achieve the same rate by designing a strand of length $200$ where we can correct one substitution error in every concatanable of length as low as $40$ (where $m =37$) with storage overhead $4$--$5$ times of what their scheme requires and lower maximum run-length. Alternatively, we design a strand of length $400$ where one substitution error in every concatanable segment of length $80$ (where $m=77$) is corrected, achieving a projected rate of $1.6250$ provided that the system can afford a notably higher storage. 
\par
Both coding schemes have quadratic time-complexity and satisfy the biochemical constraints (homopolymer and $GC$-content) required by the storage system. In fact, our DNA strands do not contain homopolymers of length $3$, which may improve error rates, and also we detect/correct double-substitution errors with high probability (see the next section). Additional properties of EC D-LOCO codes include local balance and parallel encoding/decoding (see \cite[Section~VII C--E]{D-LOCO}). A crucial advantage of our scheme is that our design is not limited by strand length, due to its concatenated structure, as long as the adopted sequencer offers sufficient read length. This limits the rate loss due to the labeling of strands. By employing a marker-based reconstruction, our scheme stands as a potential candidate for nanopore sequencer, where the error rates are higher than for its alternatives. See~\cite[Table~I]{jo_NGScomp} for a comparison of next-generation sequencing (NGS) methods. We leave this direction for future work.

Song et al.~\cite{song_delsub} construct $4$-ary codes of length $m$ correcting $s$ substitutions with asymptotic redundancy $(4s-1-\lfloor \frac{2s-1}{4} \rfloor) \log_4\,(m) + O(\log_4\,(m))$. In particular, for $s=1$, their codes have $\frac{3}{2} \log_2(m)+O(\log_2(m))$ redundancy with $O(m^3)$ (resp., $O(m^2)$) systematic encoding (resp., decoding) complexities. EC D-LOCO codes, on the other hand, are constrained, i.e., balanced and free of long homopolymers, and have encoding/decoding algorithms with quadratic, i.e., $O(m^2)$, complexities as well as a list-decoding procedure for double-substitutions that requires at most $O(m^3)$ time. 

Bar-Lev et al.~\cite{barlev_etal} offer a modular scheme that combines a synchronization algorithm based on deep neural networks with an error-correction constrained coding scheme relying on an inner-outer coding approach with additional novel procedures and sophisticated steps. In one of their outer (constrained) codes, they encode $13$-bit data chucks into DNA blocks of length $7$. Imposing the restriction that these blocks do contain runs of length $5$ in the middle and do not contain runs of length $3$ at the edges, they design DNA sequences with maximum run-length $= 4$. In the concatenation process, they address balancing as well and make sure that the encoded sequences have guaranteed $GC$-content $30\%$--$70\%$ ($45\%$--$55\%$ with high probability). As part of the encoding stage, constrained codes are followed by a Bose-Chaudhuri-Hocquenghem (BCH) code as the inner code that can correct $3$ substitutions, with which they achieve an average code rate of $1.6$, taking the rate loss due to indexing/labeling into account. They use syndrome decoding for the BCH code, requiring a storage overhead of few hundreds of megabytes. Despite its advantages, their scheme cannot fully address the biochemical constraints due to the unconstrained segment of parity-check symbols as well as a lack of proper bridging with the (constrained) labels. Our storage overhead is remarkably less than that of their scheme.

\subsection{How to address double-substitution errors} \label{sec:correcting_double_errors}
Consecutive (i.e., burst-type) substitution errors of length $2$ have been shown to occur during DNA synthesis with widely used technologies, and therefore these errors persist in the retrieved DNA sequence even under ideal storage and sequencing \cite{heckel_twin}. In addition to this barrier to reliable decoding, the multiple-read approach of DNA strands, apart from its efficiency drawbacks, cannot be sufficient with limited number of noisy reads (even at low sequencing error rates) to prevent double-substitution errors from showing up in the reconstructed codewords. For these reasons, it is important to address double-substitution errors with high decoding success. We introduce a list-decoding procedure, making use of our residue table, for double-substitution errors, which runs in near-quadratic time in $m$ and at no rate cost. By incorporating this procedure, we complete the residue decoding algorithm under Assumption~1 or Assumption~2 below. Without loss of generality, experimental results below are for $\ell=2$ and lengths $m \in \{23,41,55,61\}$ as well as for $\ell = 3$ and $m \in \{17,23\}$. \\
\textbf{Assumption 2.} Suppose two substitution errors occurred at some symbols of the EC codeword $\mathbf{d}^{\mathrm{c}}$ but not in the following three bridging symbols. Denote the erroneous word by $\mathbf{w}$. 

Let $P_1$ be the probability of $\mathbf{w}$ being a codeword and having a formal index congruent to $0$ or $N(m)-1$, i.e.,
\begin{align} P_1 = \mathbb{P} [\,(g(\mathbf{w}&, \ell) \equiv 0 \textrm{ } (\mathrm{mod} \textup{ } R)) \nonumber \\
& \vee \textrm{ } (g(\mathbf{w}, \ell) \equiv N(m)-1 \textup{ } (\mathrm{mod} \textup{ } R))\,].
\end{align}

\noindent Let $P_2$ be the probability of $\mathbf{w}'$ being in $\mathcal{D}^{\mathsf{Res}}_{m,\ell}$ with the correct check-sum and having Hamming distance $1$ to $\mathbf{w}$, i.e., after one substitution is applied to $\mathbf{w}$, which results in $\mathbf{w}'$ of Hamming distance at most $3$ to $\mathbf{d}^{\mathrm{c}}$.

Experimental results show that $P_1+P_2 < 0.0005$, i.e., such a word $\mathbf{w}'$ leaves the purple diamond with a ``No'' with probability at least $99.95\%$, allowing us to reach the conclusion that the inputted word $\mathbf{w}$ has 
double-substitution errors. The upper bound for $P_1 + P_2$ applies regardless of whether the encoding codeword is initially 
complemented or not due to EC balancing. Here are the algorithm updates.

After this ``No'', complement $\mathbf{w}$ back, if needed, based on the trusted first bridging symbol to get $\mathbf{w}^{\mathrm{c}}$, and then follow up with the list-decoding step below, namely Step~D4.
  
\begin{itemize}
\item[D4.] 
\begin{itemize}
\item[(1)] List all words $\mathbf{u}$ that are of Hamming distance $1$ to the word $\mathbf{w}^{\mathrm{c}}$ (simply by introducing all single substitutions to $\mathbf{w}^{\mathrm{c}}$), at least one of which is of Hamming distance $1$ to $\mathbf{d}$.
\item[(2)] For each listed word $\mathbf{u}$, find its formal residual index $(g(\mathbf{u},\ell)\textup{ } \mathrm{mod} \textup{ } R)$.
\item[(3)] For each listed $\mathbf{u}$, if codeword, apply Step~OD2; otherwise apply Case 2 in Step~OD1. If the resulting word $\mathbf{u}'$ is an EC codeword with the correct check-sum (equal to that of $\mathbf{d}$) and of Hamming distance $2$ to $\mathbf{w}^{\mathrm{c}}$, include it in the decoding list $\mathcal{L}$. Proceed with Step~D5.
\end{itemize}
\end{itemize}
\begin{itemize}
\item[D5.] If $|\mathcal{L}|=1$, output the index $g(\mathbf{u}')$ of the unique $\mathbf{u}' \in \mathcal{L}$. If not, output one element index in $\mathcal{L}$ at random. 
\end{itemize}
Note that $\mathbf{d}$ must be in $\mathcal{L}$. Note also that the above procedure is more likely to succeed (at a lower cost) than listing~(brute force) all D-LOCO codewords that are (i) of Hamming distance $2$ to $\mathbf{w}^{\mathrm{c}}$, (ii) with the check-sum equal to that of $\mathbf{d}$, and (iii) with zero residual index. 

The \textit{relative} time-complexity of the list-decoding step to the correction of single substitutions is $O(m)$ as up to $3m$ binary searches and index corrections, which are of $O(m^2)$, are required. In particular, the complexity of Step~D4 is $O(m\cdot m\log(R(m)))$, which is $O(m^2\log(m))$, based on the polynomial modeling of $R(m)$. The complexity of Step~D5 depends on the behavior of the cardinality $|L|$ as $m \rightarrow \infty$. The overall asymptotic complexity of our list-decoding procedure is at most $O(m^3)$, if $|L|=O(m)$, as then D5 dominates D4 in terms of complexity, and is expected to be $O(m^2 \cdot \log(m))$, if $|L|=O(1)$, as then D4 dominates instead.

The conditional probability of brute-force listing for a random $\mathbf{w}$ is experimentally found. In the experiment, we pick random binary data and apply double-substitution errors to the encoding EC codeword $\mathbf{d}$ in a manner that the erroneous word $\mathbf{w}$ is not an EC codeword. We then list all EC codewords that are (i) of Hamming distance $2$ to $\mathbf{w}$ and (ii) with the check-sum equal to the local detection check-sum of $\mathbf{d}$. This is repeated $5000$ ($4000$) times for $m \in \{23, 41, 55\}$ ($m=61$) in case $\ell = 2$, and $10000$ times for $m \in \{17, 23\}$ in case $\ell = 3$. Since Step~D4 is more likely to succeed than brute-force listing, the success rate of our list-decoding procedure is at least $98.25\%$ ($99.4\%$) for $\ell = 2$ ($\ell = 3$) independent of the specified lengths $m$. In fact, these percentages go up to $99.1\%$ and $99.7\%$, respectively, after random outputting as the lists of size $2$ are experimentally dominant in the case of failed decoding for both $\ell=2,3$. 


\begin{figure}[ht!]
\vspace{-1.25em}
\center
\includegraphics[trim={0.0in 0.0in 0.0in 0.0in}, width=3.6in]{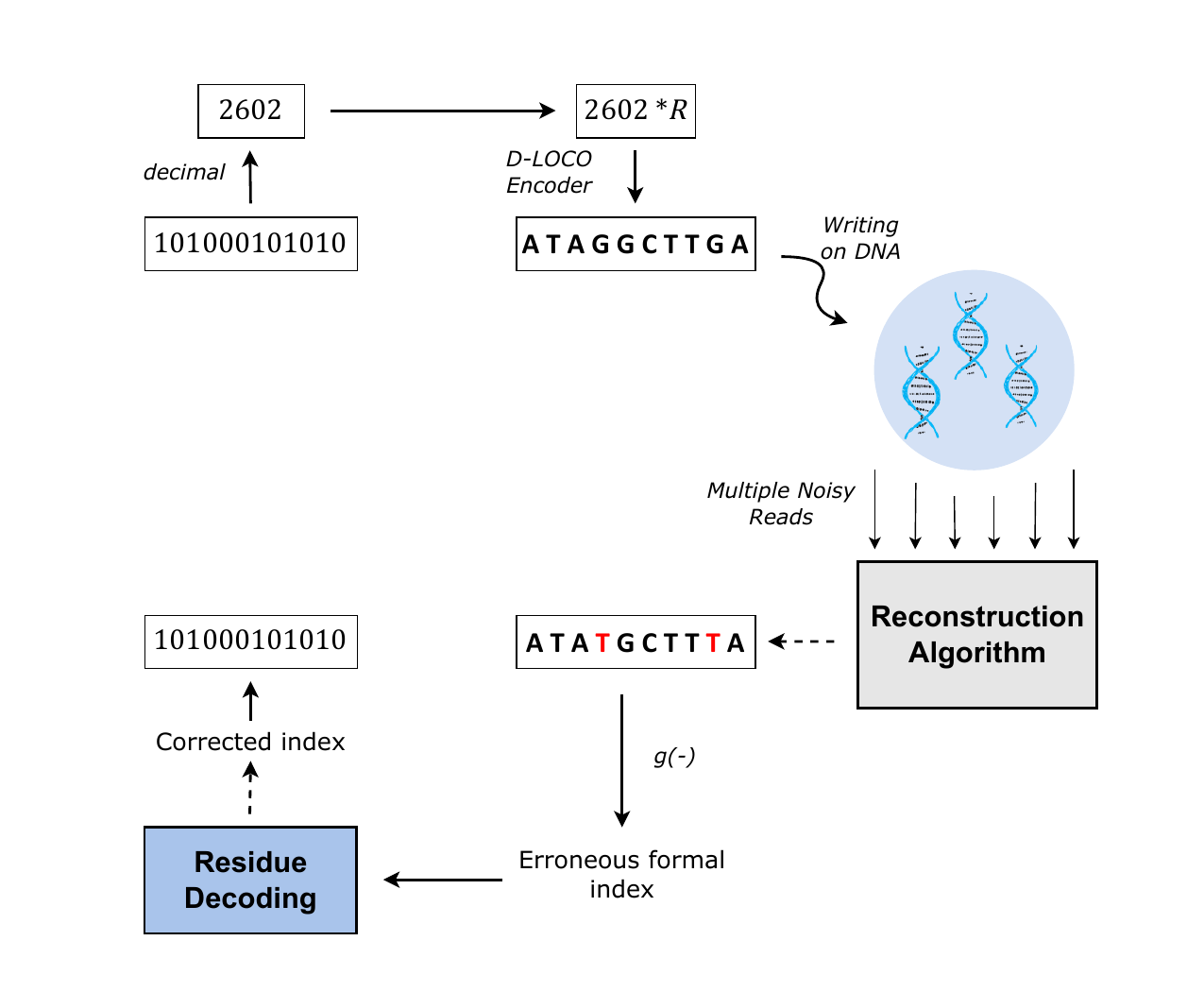}
\vspace{-2em}
\caption{Panorama of the data retrieval process to illustrate how our error-correction algorithm would fit into it once preceded by a reconstruction algorithm. The dashed line following the reconstruction algorithm indicates an outcome of correct length with high probability, which our decoder effectively handles.}
\label{fig:pano}
\vspace{-0.5em}
\end{figure}

\section{Potential of Our Scheme and Future Directions}

Our coding scheme is to be completed by adopting a reconstruction algorithm in order to handle deletions and insertions as well, which also follows the literature trend of multiple reads of DNA strands succeeded by a reconstruction process. Fig.~\ref{fig:pano} depicts a block diagram of the system we assume. A crucial property required from a suitable reconstruction algorithm to complement EC D-LOCO codes is the efficacy to recover a strand from its multiple noisy reads such that
\begin{itemize}
\item[(a)] substitutions are allowed but with correct length (as deletions/insertions are becoming new substitutions), and
\item[(b)] the efficiency of the overall scheme should not be negatively affected. 
\end{itemize}

This problem is, however, beyond the scope of this paper, and we instead cite relevant literature works \cite{milenkovic_trace}, \cite{sabary_rec}, \cite{barlev_etal}. By employing a marker-based version of available reconstruction algorithms in the literature (based on using more bridging symbols as markers), our scheme would enable the reconstruction of the (possibly erroneous) EC codewords in a parallel manner. This parallelism offers low latency, while reconstruction of individual codewords (in place of the whole strand at once) decreases the time-complexity of the adopted reconstruction algorithm that depends on the length of the input sequence (and the number of noisy reads as well). Handling substitution errors effectively via our proposed codes can allow us to prioritize deletions and insertions during the reconstruction process due to (a), and thus achieve full recovery of data with fewer number of reads, also decreasing the time-complexity of decoding. Such potential update and usage of our scheme is to be assessed in future work.\footnote{Near-future tasks are to search for ways to further maximize EC D-LOCO rates and minimize their storage overhead as well as to further polish the analysis of success rates for correcting double substitutions (occurring within strand segments). A theoretical analysis of the asymptotic behavior of the redundancy metric $R(m)$ is an additional key action item too.}

\section{Conclusion}

We introduced EC D-LOCO codes as subcodes of D-LOCO codes with guaranteed correction of single-substitution errors and probabilistic (highly likely) correction of double-substitution errors by adding redundancy via using only D-LOCO codewords with indices having certain property. The exhaustive set of index differences due to single substitutions was derived for generic code lengths and a fixed maximum run-length in $\{1,2,3\}$. The adopted redundancy metric for guaranteed correction of single-substitution errors was obtained via this set, and its modeling led to the projection that EC D-LOCO codes approach the constrained-system capacity. Future work includes further developing our coding solution so that it can address deletion and insertion errors as well.


\section*{Acknowledgment}

The authors would like to thank Yusuf Uslan and Arda D\"uzta\c{s} for their assistance in carrying out this research. Furthermore, they would like to thank the Associate Editor Prof. Reza Asvadi for his effective handling of the article, and thank the anonymous reviewers for their constructive comments.




\begin{thebibliography}{13}

\balance

\bibitem{goldman}
{
N. Goldman et al., ``Towards practical, high-capacity, low-maintenance information storage in synthesized DNA,''  \textit{Nature}, vol. 494, pp. 77--80, Jan. 2013.
}

\bibitem{grass_robust}
{
R. N. Grass, R. Heckel, M. Puddu, D. Paunescu, and W. J. Stark, ``Robust chemical preservation of digital information on DNA in silica with error-correcting codes,'' \textit{Angew. Chem. Int. Ed.}, vol. 54, no. 8, pp. 2552--2555, Feb. 2015.
}

\bibitem{blawat}
{
M. Blawat et al., ``Forward error correction for DNA data storage,'' \textit{Procedia Comput. Sci.}, vol. 80, pp. 1011--1022, Jun. 2016.}

\bibitem{erlich_fountain}
{ 
Y. Erlich and D. Zielinski, ``DNA Fountain enables a robust and efficient storage architecture,'' \textit{Science}, vol. 355, no. 6328, pp. 950--954, Mar. 2017.
}

\bibitem{organick_etal}
L. Organick et al., ``Random access in large-scale DNA data storage,'' \textit{Nat. Biotechnol.}, vol. 36. no. 3, pp. 242--258, Mar. 2018.


\bibitem{wang_etal19}
{
Y. Wang et al., ``High capacity DNA data storage with variable-length oligonucleotides using repeat accumulate code and hybrid mapping,'' \textit{J. Biol. Eng.}, vol. 13, pp. 1--11, Nov. 2019.
}

\bibitem{milenkovic_rw2}
{
C. Pan et al., ``Rewritable two-dimensional DNA-based data storage with machine learning reconstruction,'' \textit{Nat. Commun.}, vol. 13, no. 2984, May 2022.
}


\bibitem{DNA_engine}
{
K. N. Lin et al., ``A primordial DNA store and compute engine,'' \textit{Nat. Nanotechnol.}, vol. 19, pp. 1654--1664, Aug. 2024.
}


\bibitem{ding_etal}
{
Y. Ding et al., ``High information density and low coverage data storage in DNA with efficient channel coding schemes,'' 2024, \textit{arXiv:2410.04886}.
}


\bibitem{DNAprint}
{ 
C. Zhang et al., ``Parallel molecular data storage by printing epigenetic bits on DNA,'' \textit{Nature}, vol. 634, pp. 824--832, Oct. 2024.
}

\bibitem{cai_optimal}
{
K. Cai, X. He, H. M. Kiah, and T. Thanh Nguyen, ``Efficient constrained encoders correcting a single nucleotide edit in DNA storage,'' in \textit{ Proc. Int. Conf. Acoust., Speech, Signal Process. (ICASSP)}, Barcelona, Spain, May 2020, pp. 8827--8830.
}

\bibitem{nguyen_etal} T. T. Nguyen, K. Cai, K. A. S. Immink, and H. M. Kiah, ``Capacity-approaching constrained codes with error correction for DNA-based data storage,'' \textit{IEEE Trans. Inf. Theory}, vol. 67, no. 8, pp. 5602--5613, Aug. 2021.

\bibitem{ross_etal}  M. G. Ross et al., ``Characterizing and measuring bias in sequence data,'' \textit{Genome Biol.}, vol. 14, no. 5, p. R51, May 2013.

\bibitem{schwartz_etal} J. J. Schwartz, C. Lee, and J. Shendure, ``Accurate gene synthesis with tag-directed retrieval of sequence-verified DNA molecules,'' \textit{Nat. Methods}, vol. 9, no. 9, pp. 913--915, Aug. 2012.

\bibitem{olgica_nano}
O. Milenkovic, R. Gabrys, H. M. Kiah, and S. M. H. Tabatabaei Yazdi, ``Exabytes in a test tube,'' \textit{IEEE Spectr.}, vol. 55, no. 5, pp. 40--45, May 2018.

\bibitem{heckel_etal}
R. Heckel, G. Mikutis, and R. N. Grass, ``Characterization of the DNA data storage channel,'' \textit{Sci. Rep.}, vol. 9, no. 9663, Jul. 2019, doi.org/10.1038/s41598-019-45832-6.

\bibitem{immink_cai}
K. A. S. Immink and K. Cai, ``Design of capacity-approaching constrained codes for DNA-based storage systems,'' \textit{IEEE Commun. Lett.}, vol. 22, no. 2, pp. 224--227, Feb. 2018.

\bibitem{hedges}
{
W. H. Press et al., ``HEDGES error-correcting code for DNA storage corrects indels and allows sequence constraints,'' in \textit{Proc. Natl. Acad. Sci. USA}, vol. 117, no. 32, pp. 18489--18496, Aug. 2020.
}

\bibitem{wu_1}
X. Li, M. Chen, and H. Wu, ``Multiple errors correction for position-limited DNA sequences with GC balance and no homopolymer for DNA-based data storage,'' \textit{Brief. Bioinform.}, vol. 24, no. 1, Jan. 2023. 

\bibitem{dna_aeon}
{
M. Welzel et al., ``DNA-Aeon provides flexible arithmetic coding for constraint adherence and error correction in DNA storage,'' \textit{Nat. Commun.}, vol 14, no. 628,  Feb. 2023.}

\bibitem{wu_2}
Z. Yan, G. Qu, and H. Wu, ``A novel soft-in soft-out decoding algorithm for VT codes on multiple received DNA strands,'' in \textit{Proc. IEEE Int. Symp. Inf. Theory (ISIT)}, Taipei, Taiwan, Jun. 2023, pp. 838--843.

\bibitem{wu_3}
Z. Yan, C. Liang, and H. Wu, ``A Segmented-edit error-correcting code with re-synchronization function for DNA-based storage systems,'' \textit{IEEE Trans. Emerg. Top. Comput.}, vol. 11, no. 3, pp. 605--618,  Jul.-Sep. 2023.

\bibitem{parkpark} 
{
S. -J. Park, H. Park, H. -Y. Kwak, and J. -S. No, ``BIC codes: Bit insertion-based constrained codes with error correction for DNA storage,'' \textit{IEEE Trans. Emerg. Topics Comput.}, vol. 11, no. 3, pp. 764--777, Jul.-Sep. 2023.
}

\bibitem{chandak_nano}
{
S. Chandak et al., ``Overcoming high nanopore basecaller error rates for DNA storage via basecaller-decoder integration and convolutional codes,'' in \textit{Proc. Int. Conf. Acoust., Speech, Signal Process. (ICASSP),  Barcelona, Spain}, May 2020, pp. 8822--8826.
}

\bibitem{sima1}
{
J. Sima, N. Raviv, and J. Bruck, ``On coding over sliced information,'' \textit{IEEE Trans. Inf. Theory}, vol. 67, no. 5, pp. 2793--2807, May 2021.
}

\bibitem{sima2}
{
J. Sima and J. Bruck, ``On optimal k-deletion correcting codes,'' \textit{IEEE Trans. Inf. Theory}, vol. 67, no. 6, pp. 3360--3375, Jun. 2021.
}

\bibitem{cai_seqrec}
{
K. Cai, H. M. Kiah, T. T. Nguyen, and E. Yaakobi, ``Coding for sequence reconstruction for single edits,'' \textit{IEEE Trans. Inf. Theory}, vol. 68, no. 1, pp. 66--79, Jan. 2022.
}

\bibitem{song_delsub}
{
W. Song, N. Polyanskii, K. Cai, and X. He, ``Systematic codes correcting multiple-deletion and multiple-substitution errors,'' \textit{IEEE Trans. Inf. Theory}, vol. 68, no. 10, pp. 6402--6416, Oct. 2022.
}

\bibitem{xing2del} 
{
S. Liu, I. Tjuawinata, and C. Xing, ``Explicit construction of q-ary 2-deletion correcting codes with low redundancy,'' \textit{IEEE Trans. Inf. Theory}, vol. 70, no. 6, pp. 4093--4101, Jun. 2024.
}

\bibitem{serge}
{
S. K. Hanna, ``Short systematic codes for correcting random edit errors in DNA storage,'' in \textit{Proc. IEEE Int. Symp. Inf. Theory (ISIT)}, Athens, Greece, Jul. 2024, pp. 663--668.
}

\bibitem{song_burst}
{
W. Song, K. Cai, and T. Q. S. Quek, ``New construction of q-ary codes correcting a burst of at most t deletions,'' in \textit{Proc. IEEE Int. Symp. Inf. Theory (ISIT)}, Athens, Greece, Jul. 2024, pp. 1101--1106.
}

\bibitem{burst_ge}
{
Y. Sun and G. Ge, ``Codes for correcting a burst of edits using weighted-summation VT sketch,'' \textit{IEEE Trans. Inf. Theory}, vol. 71, no. 3, pp. 1631--1646, Mar. 2025.
}

\bibitem{substring_farnoud}
{
Y. Li, Y. Tang, H. Lou, R. Gabrys, and F. Farnoud, ``Optimal codes correcting a substring edit,'' \textit{IEEE Trans. Inf. Theory}, early access, Apr. 2025, doi: 10.1109/TIT.2025.3562730.
}

\bibitem{wang_etal}
Y. Wang, M. Noor-A-Rahim, E. Gunawan, Y. L. Guan, and C. L. Poh, ``Construction of bio-constrained code for DNA data storage,'' \textit{IEEE Commun. Lett.}, vol. 23, no. 6, pp. 963--966, Jun. 2019.

\bibitem{park_lee_no}
S. -J. Park, Y. Lee, and J. -S. No, ``Iterative coding scheme satisfying GC balance and run-length constraints for DNA storage with robustness to error propagation,'' \textit{J. Commun. Netw.}, vol. 24, no. 3, pp. 283--291, Jun. 2022.

\bibitem{he_liu_wang_tang}
{
X. He, Y. Liu, T. Wang, and X. Tang, ``Efficient explicit and pseudo-random constructions of constrained codes for DNA storage,'' \textit{IEEE Trans. Commun.}, vol. 73, no. 3, pp. 1431--1443, Mar. 2025.
}

\bibitem{D-LOCO} 
{C. \.{I}rima\u{g}z{\i}, Y. Uslan, and A. Hareedy, ``Protecting the future of information: LOCO coding with error detection for DNA data storage,'' \textit{IEEE Trans. Mol. Biol. Multi-Scale Commun.}, vol. 10, no. 2, pp. 317--333, Jun. 2024.
}

\bibitem{ahh_general}
A. Hareedy, B. Dabak, and R. Calderbank, ``The secret arithmetic of patterns: A general method for designing constrained codes based on lexicographic indexing,'' \textit{IEEE Trans. Inf. Theory}, vol. 68, no. 9, pp. 5747--5778, Sep. 2022.

\bibitem{const_ONT}
K. Whritenour, M. Civelek, and F. Farnoud, ``Constrained code for data storage in DNA via nanopore sequencing,'' in \textit{Proc. Asilomar Conf. Signals, Syst. Comput.}, Pacific Grove, CA, USA, Oct. 2023, pp. 975--981.

\bibitem{sima_syndrome}
{
J. Sima, R. Gabrys, and J. Bruck, ``Syndrome compression for optimal redundancy codes'' in \textit{Proc. IEEE Int. Symp. Inf. Theory (ISIT)}, Los Angeles, CA, USA, Jun. 2020, pp. 751--756.
}

\bibitem{VT}
{
R. R. Varshamov and G. M. Tenengolts, ``Codes which correct single asymmetric errors,'' \textit{Avtomatika Telemekhanika}, vol. 26, no. 2, pp. 288--292, 1965.
}

\bibitem{ahh_loco}
A. Hareedy and R. Calderbank, ``LOCO codes: Lexicographically-ordered constrained codes,'' \textit{IEEE Trans. Inf. Theory}, vol. 66, no. 6, pp. 3572--3589, Jun. 2020.


\bibitem{bd_tdmr}
B. Dabak, A. Hareedy, and R. Calderbank, ``Non-binary constrained codes for two-dimensional magnetic recording,'' \textit{IEEE Trans. Magn.}, vol. 56, no. 11, pp. 1--10, Nov. 2020.

\bibitem{ahh_mdl}
A. Hareedy, B. Dabak, and R. Calderbank, ``Managing device lifecycle: Reconfigurable constrained codes for M/T/Q/P-LC Flash memories,'' \textit{IEEE Trans. Inf. Theory}, vol. 67, no. 1, pp. 282--295, Jan. 2021.


\bibitem{ahh_ps}
A. Hareedy, S. Zheng, P. Siegel, and R. Calderbank, ``Efficient constrained codes that enable page separation in modern Flash memories,'' \textit{IEEE Trans. Commun.}, vol. 71, no. 12, pp. 6834-6848, Dec. 2023.

\bibitem{dua_tdmr}
{ 
D. \"{O}zbayrak, D. Uyar, and A. Hareedy, ``Low-complexity constrained coding schemes for two-dimensional magnetic recording,'' in \textit{Proc. IEEE Int. Symp. Inf. Theory (ISIT)}, Athens, Greece, Jul. 2024, pp. 825--830.}

\bibitem{goca_tdmr}
{
I. Guzel, D. \"{O}zbayrak, R. Calderbank, and A. Hareedy, ``Eliminating media noise while preserving storage capacity: Reconfigurable constrained codes for two-dimensional magnetic recording,''  \textit{IEEE Trans. Inf. Theory}, vol. 70, no. 7, pp. 4905--4927, Jul. 2024
}

\bibitem{barlev_etal}
{
D. Bar-Lev et al., ``Scalable and robust DNA-based storage via coding theory and deep learning,'' \textit{Nat. Mach. Intell.}, vol. 7, pp. 639--649, Feb. 2025.
}

\bibitem{jo_NGScomp}
S. Jo et al., ``Recent progress in DNA data storage based on high-throughput DNA synthesis,'' \textit{Biomed. Eng. Lett.}, vol. 14, no. 5, pp. 993--1009, Sep. 2024.


\bibitem{heckel_twin}
{ 
A. L. Gimpel, W. J. Stark, R. Heckel, and R. N. Grass, ``A digital twin for DNA data storage based on comprehensive quantification of errors and biases,'' \textit{Nat. Commun.}, vol. 14, no. 6026, Sep. 2023. 
}

\bibitem{milenkovic_trace}
{
M. Cheraghchi, R. Gabrys, O. Milenkovic, and J. Ribeiro, ``Coded trace reconstruction'' \textit{IEEE Trans. Inf. Theory}, vol. 66, no. 10, pp. 6084--6103, Oct. 2020.
}

\bibitem{sabary_rec}
{
O. Sabary, A. Yucovich, G. Shapira, and E. Yaakobi, ``Reconstruction algorithms for DNA-storage systems,'' \textit{Sci. Rep.}, vol. 14, pp. 1951, Jan. 2024.
}



\end{thebibliography}
\end{document}